\documentclass[USletter,10pt]{article}
\usepackage[dvips]{geometry}
\usepackage{graphicx,cancel}
\usepackage{graphics}
\usepackage{color, soul}
\usepackage{amsmath, amsthm, slashed}
\usepackage{amssymb}
\usepackage{rotating}
\usepackage{mathrsfs}
\usepackage{epsfig}
\usepackage{verbatim}
\usepackage[affil-it]{authblk}
\usepackage{rotating}
\usepackage{graphicx}
\usepackage{accents}
\newtheorem{definition}{Definition}[section]
\newtheorem{theorem}{Theorem}[section]
 
\newtheorem*{conjecture*}{Conjecture}
\newtheorem{corollary}{Corollary}[section]
\newtheorem*{theorem*}{Theorem}
\newtheorem*{corollary*}{Corollary}
\newtheorem{proposition}{Proposition}[subsection]
\newtheorem{lemma}{Lemma}[subsection]
\newtheorem{remark}{Remark}[section]

  \newtheoremstyle{named}{}{}{\itshape}{}{\bfseries}{.}{.5em}{\thmnote{#3 }#1}
\theoremstyle{named}
\newtheorem*{namedtheorem}{Theorem}

\newcommand{\bea}{\begin{eqnarray}}
\newcommand{\eea}{\end{eqnarray}}
\def\beaa{\begin{eqnarray*}}
\def\eeaa{\end{eqnarray*}}
\def\ba{\begin{array}}
\def\ea{\end{array}}
\def\be#1{\begin{equation} \label{#1}}
\def \eeq{\end{equation}}
\def\bsplit{\begin{split}}
\newcommand{\nn}{\nonumber}

\def\les{\lesssim}

\def\c{\cdot}

\def\dual{{\,^\star \mkern-2mu}}
\def\tr{\mbox{tr}}

\newcommand{\nabb}{\nab\mkern-13mu /\,}

\newcommand{\curl}{\mbox{curl }}
\newcommand{\divv}{\mbox{div}\mkern-19mu /\,\,\,\,}
\newcommand{\lapp}{\mbox{$\bigtriangleup  \mkern-13mu / \,$}}
\newcommand{\curll}{\mbox{curl}\mkern-19mu /\,\,\,\,}

\def\nab{\nabla}

\def\pr{\partial}

\def\DDs{ \, \DD \hspace{-2.4pt}\dual    \mkern-16mu /}
\def\DDd{ \, \DD \hspace{-2.4pt}    \mkern-8mu /}

\def\a{\alpha}
\def\b{\beta}

\def\de{\delta}

\def\ep{\epsilon}

\def\om{\omega}

\def\th{{\theta}}

\def\ka{\kappa}
\def\ze{\zeta}

\def\vphi{{\varphi}}

\renewcommand{\aa}{\protect\underline{\a}}
\newcommand{\bb}{\protect\underline{\b}}
\def\omb{{\underline{\om}}}

\newcommand{\chib}{\underline{\chi}}
\newcommand{\xib}{\underline{\xi}}
\newcommand{\etab}{\underline{\eta}}

\def\kab{\underline{\kappa}}

\def\DD{{\mathcal D}}

\def\MM{{\mathcal M}}

\def\D{{\bf D}}

\def\F{{\bf F}}

\def\R{{\bf R}}

\def\W{{\bf W}}

\def\g{{\bf g}}

\def\fb{\underline{f}}

\def\pf{\frak{p}}
\def\qf{\frak{q}}

\def\ff{\frak{f}}

\def\chih{\widehat{\chi}}
\def\chibh{\widehat{\chib}}

\def\bF{\,^{(F)} \hspace{-2.2pt}\b}
\def\bbF{\,^{(F)} \hspace{-2.2pt}\bb}

\def\rhoF{\,^{(F)} \hspace{-2.2pt}\rho}

\def\sigmaF{\,^{(F)} \hspace{-2.2pt}\sigma}

\def\c{\cdot}
\def \f12{\frac 1 2 }
\def\ov{\overline}

\DeclareMathAlphabet\mathbfcal{OMS}{cmsy}{b}{n}

\title{Boundedness and decay for the Teukolsky equation of spin $\pm1$ on Reissner-Nordstr{\"o}m spacetime: the $\ell=1$ spherical mode}

\author{Elena Giorgi}

\date{}

\begin{document}

\maketitle

\begin{abstract}
We prove boundedness and polynomial decay statements for solutions to the spin $\pm1$ Teukolsky-type equation projected to the $\ell=1$ spherical harmonic on Reissner-Nordstr{\"o}m spacetime. The equation is verified by a gauge-invariant quantity which we identify and which involves the electromagnetic and curvature tensor. This gives a first description in physical space of gauge-invariant quantities transporting the \textit{electromagnetic radiation} in perturbations of a charged black hole.  

The proof is based on the use of derived quantities, introduced in previous works on linear stability of Schwarzschild (\cite{DHR}).  The derived quantity verifies a Fackerell-Ipser-type equation, with right hand side vanishing at the $\ell=1$ spherical harmonics. The boundedness and decay for the projection to the $\ell\geq 2$ spherical harmonics are implied by the boundedness and decay for the Teukolsky system of spin $\pm2$ obtained in \cite{Giorgi4}.

The spin $\pm1$ Teukolsky-type equation is verified by the curvature and electromagnetic components of a gravitational and electromagnetic perturbation of the Reissner-Nordstr{\"o}m spacetime. Consequently, together with the estimates obtained in \cite{Giorgi4}, these bounds allow to prove the full linear stability of Reissner-Nordstr{\"o}m metric for small charge to coupled gravitational and electromagnetic perturbations. 
\end{abstract}

  \tableofcontents
  
\section{Introduction}
The problem of stability of the Kerr family as solution to the Einstein vacuum equation is one the main open problems in General Relativity. The equivalent problem in the setting of electrovacuum spacetimes is the stability of the Kerr-Newman family as solution to the Einstein-Maxwell equations. 

 The Reissner-Nordstr{\"o}m family
  of spacetimes $(\mathcal{M},g_{M, Q})$ is the simplest non-trivial solution of the Einstein-Maxwell equation. It can be   expressed  in local coordinates as
  \begin{equation}
  \label{RNintro}
 g_{M, Q}=-\left(1-\frac{2M}{r}+\frac{Q^2}{r^2}\right)dt^2 +\left(1-\frac{2M}{r}+\frac{Q^2}{r^2}\right)^{-1}dr^2 +r^2(d\theta^2+\sin^2\theta d\phi^2),
  \end{equation}
where $M$ and $Q$ are two real parameters verifying $|Q|<M$.

The ultimate goal of the resolution to the problem of non-linear stability of a solution of the Einstein equation consists in showing that it is stable to small perturbations of initial data as a solution to the fully non-linear Einstein equation. In particular, a complete resolution would need to prove that small perturbations of a solution remain bounded for all times, and converge to another member of the family of solutions.  Intermediate steps are the formal mode analysis of the linearized Einstein equation and the proof of decay for the physical space linearized equation. 

The linearized gravity around a solution can be formally decomposed into modes, and such decomposition in modes  allows one to prove  what is known as mode stability, i.e the lack of exponentially growing modes for all metric or curvature components.  Extensive literature by the physics community covers the formal study of fixed modes from the point of view of metric perturbations and of Newman-Penrose formalism.

The original approach to mode stability of the Reissner-Nordstr{\"o}m spacetime are the metric pertubations, leading to a generalization of the Regge-Wheeler and Zerilli equations. In his studies of the metric perturbations of Reissner-Nordstr\"om spacetime in \cite{Moncrief1}, \cite{Moncrief2}, \cite{Moncrief3}, Moncrief reduced the governing equations to two pairs of one-dimensional wave equations which govern the odd and the even-parity perturbations. In \cite{Chandra-RN} and \cite{Chandra-RN2}, Chandrasekhar showed that the two pairs are related to each other: the solutions to one parity can be deduced from the solutions to the opposite parity.  See also \cite{Xanta} and \cite{Zerilli}. 

 Analysis of the fixed mode perturbations of the Reissner-Nordstr{\"o}m black hole via the Newman-Penrose formalism were treated in \cite{Chandra1} and \cite{Bicak}, where the corresponding Teukolsky equations are derived. In \cite{Chandra-RN2}, Chandrasekhar showed that the Newman-Penrose equations can be transformed to one dimensional wave equation appropriate for the odd and even parity perturbations, now called the (fixed-frequency) Chandrasekhar transformation. 
 
 Numerical works have supported the study of mode stability in Kerr-Newman and Reissner-Nordstr\"om spacetime, see for example \cite{numerics1} and \cite{numerics2}. Mode stability for charged spacetimes in higher dimensions has been treated in \cite{highdim}.

All the above results rely on the derivation of the equations in separated forms and are sufficient to prove that there are no exponentially growing modes, and to obtain control on the quasi-normal modes of the perturbations. However, this weak version of stability is far from sufficient to prove boundedness and decay of the solution even to the linearized Einstein equation. In particular, the behavior of the quasi-normal modes in the asymptotically flat regime is not sufficient to deduce any information on the solution in physical space. One needs instead to derive sufficiently strong decay estimates through a physical space analysis in order to be able to apply them in the nonlinear framework.

 In \cite{Giorgi4}, we started a program towards the proof of (physical space, non-modal) linear stability of the Reissner-Nordstr{\"o}m spacetime to coupled gravitational and electromagnetic perturbations. See the introduction of \cite{Giorgi4} for an introduction on the problem of stability of charged black holes.  Since we perturb the curvature tensor using null frames, our analysis is related to the one in Newman-Penrose formalism, but will not make any use of decomposition in modes. The analysis is entirely in physical space, and no choice of gauge are performed in this paper. We will instead treat gauge-invariant quantities, as defined in Section \ref{gauge-inv-section}.
 
 The main result in \cite{Giorgi4} is the proof of boundedness and decay for two symmetric traceless two tensors $\a$ and $\ff$,  governing the gravitational radiation of the Reissner-Nordstr{\"o}m spacetime. The decay for $\a$ and $\ff$, veryfing Teukolsky-type equations of spin $2$, was obtained by applying a  physical space version of the Chandrasekhar transformation to both $\a$ and $\ff$. This transformation consists in taking two null derivatives of $\a$ and one derivative of $\ff$ to obtain a system of generalized Regge-Wheeler equations for which combined estimates were derived.

 In the linear stability of Schwarzschild spacetime to gravitational perturbations in \cite{DHR}, the decay for $\a$ implies the decay of all the other curvature components and Ricci coefficients supported in $\ell\geq2$ spherical harmonics. 
 In addition, an intermediate step of the proof is the following theorem: Solutions of the linearized gravity around Schwarzschild supported only on $\ell=0,1$ spherical harmonics are a linearized Kerr plus a pure gauge solution.

In the setting of linear stability of Reissner-Nordstr{\"o}m to coupled gravitational and electromagnetic perturbations, we expect to have electromagnetic radiation supported in $\ell \geq 1$ spherical mode, as for solutions to the Maxwell equations in Schwarzschild (see \cite{Blue} or \cite{Federico}).

 In particular, the decay for the two tensors $\a$ and $\ff$ obtained in \cite{Giorgi4} will not give any decay information about the $\ell=1$ spherical mode of the perturbations. It turns out that, in the case of solutions to the linearized gravitational and electromagnetic perturbations around Reissner-Nordstr{\"o}m spacetime, the projection to the $\ell=0,1$ spherical harmonics is not exhausted by the linearized Kerr-Newman and the pure gauge solutions. Indeed, the presence of the Maxwell equations involving the extreme curvature component of the electromagnetic tensor, which is a one-form, transports electromagnetic radiation supported in $\ell\geq 1$ spherical harmonics. 
 
In \cite{Moncrief3}, Moncrief first analyzes the case of mode perturbations supported in $\ell=1$ spherical harmonics, in particular the ones which correspond to the electromagnetic radiation. He derives lapse and shift functions in the metric perturbation formalism which govern the perturbation in $\ell=1$.

 In \cite{Chandra}, a complex scalar quantity, written as $2 \Psi_1 \phi_1-3\phi_0 \Psi_2$ in Newman-Penrose formalism, was identified to be invariant to the first order for the infinitesimal rotations but was not used in the subsequent analysis. Indeed, it was used to show that a gauge where $\Psi_1$ and $\phi_1$ vanish identically cannot be chosen, while a gauge where $\phi_0=\phi_2=0$, the so called phantom gauge, can be chosen.  In \cite{Chandra}, the equations governing the perturbations in the Newman-Penrose formalism were written in the phantom gauge, and all the analysis was performed in such a gauge. In particular, by choosing the phantom gauge, the above quantity was being reduced to essentially a rescaled version of the curvature component $\Psi_1$.  
 
 Fixing a gauge in the derivation of the equations essentially prevents one to identify the gauge-invariant quantities, which instead ought to be interpreted as radiation and carry physical significance.
 For this reason, the gauge-independent quantities involved in the electromagnetic radiation in Reissner-Nordstr{\"o}m spacetime were not clearly identified up to this point.

 In this paper we introduce a gauge-independent one-form $\tilde{\beta}$, which is a mixed curvature and electromagnetic component (see \eqref{definition-tilde-b} for the definition). This is a tensorial version of the scalar combination identified in \cite{Chandra} and has the additional interesting property of vanishing for linearized Kerr-Newman solutions in the setting of linear stability of Reissner-Nordstr{\"o}m. In particular, we identify this quantity to be carrying the electromagnetic radiation in the coupled gravitational and electromagnetic perturbations of Reissner-Nordstr\"om spacetime.

It is remarkable that such $\tilde{\beta}$ verifies a spin $\pm1$ Teukolsky-type equation (derived in Appendix A), with non-trivial right hand side, which can be schematically written as 
\bea\label{Teukolksy-spin1}
\Box_{g_{M, Q}} \tilde{\beta}+c_1 L( \tilde{\beta})+ c_2 \underline{L}(\tilde{\beta})+ V_1 \tilde{\beta}= R.H.S.
\eea
where $L$ and $\underline{L}$ are outgoing and ingoing null vectors and the right hand side involves curvature components, electromagnetic components and Ricci coefficients. 

By applying a physical version of the Chandrasekhar transformation, we obtain a derived quantity $\mathfrak{p}$ at the level of one derivative of $\tilde{\beta}$. We define $\pf$ as (see \eqref{quantities-1})
\beaa
\pf=\frac{1}{\kab} \nabb_3(r^5\kab \ \tilde{\b})
\eeaa
where $\kab=\tr\chib$ is the trace of the second null fundamental form. Similar physical space versions of the Chandrasekhar transformations were introduced in \cite{DHR} and \cite{Federico} (see Section \ref{work-Federico} for a comparison of the derived quantities). This transformation has the remarkable property of turning the Teukolsky-type equation of spin $1$ \eqref{Teukolksy-spin1} into a Fackerell-Ipser-type\footnote{The Fackerell-Ipser equation was encountered in the study of Maxwell equations in Schwarzschild spacetime, see \cite{Federico}.} equation, with right hand side which vanishes in $\ell=1$ spherical harmonics.

A computation, carried out in Appendix B, reveals that $\pf$ verifies an equation of the schematic form:
 \begin{equation}\label{systeml1}
\Box_{g_{M, Q}}\mathfrak{p}+ V\mathfrak{p} =  \mathcal{J}
\end{equation}
where $\mathcal{J}$ is supported in $\ell\geq 2$ spherical harmonics. 

Projecting equation \eqref{systeml1} in $\ell=1$ spherical harmonics, we obtain a scalar wave equation with vanishing right hand side, for which techniques developed in \cite{lectures} and \cite{rp} can be straightforwardly applied. This proves boundedness and decay for the projection of $\mathfrak{p}$, and therefore $\tilde{\beta}$, to the $\ell=1$ spherical mode, in Reissner-Nordstr{\"o}m spacetimes with small charge. The boundedness and decay for its projection into $\ell\geq 2$ is implied by using the result for the spin $\pm2$ Teukolsky equation in \cite{Giorgi4}, for small charge.

  We emphasize that the analysis of the Teukolsky-type and the Fackerell-Ipser-type equation is obtained in physical space, and not in frequency fixed modes.

A rough version of the main result is the following. The precise statement will be given as Main Theorem in Section \ref{section-main-theorem}.

\begin{namedtheorem}[Main](Rough version)\label{main-theorem-1-rough} Let $|Q|\ll M$. Solutions $\tilde{\b}$ and $\tilde{\bb}$ to the generalized Teukolsky equation of spin $\pm 1$ on Reissner-Nordstr{\"o}m exterior spacetimes arising from regular localized initial data remain uniformly bounded and satisfy an $r^p$-weighted energy hierarchy and polynomial decay.
\end{namedtheorem}

The Main Theorem above gives motivation for identifying the gauge-invariant quantities $\tilde{\b}$ and $\tilde{\bb}$ to the electromagnetic radiation in the emission of gravitational and electromagnetic waves in a perturbation of a charged black hole. They are the electromagnetic analogue of the quantities $\a$ and $\aa$ (corresponding to $\Phi_0$ and $\Phi_4$ in Newman-Penrose formalism) which transport gravitational radiation for vacuum black holes.  The description of electromagnetic radiation supported in $\ell \geq 1$  for charged black holes has been explicitely obtained in physical space for the first time here, since it was eluded in the choice of phantom gauge chosen in \cite{Chandra}.

The Main Theorem in the present paper and the Main Theorem in \cite{Giorgi4} provide decay for the three quantities $\alpha$, $\mathfrak{f}$, $\tilde{\beta}$, and their corresponding negative spin versions. These decays imply boundedness and decay of all the remaining quantities in the linear stability for coupled gravitational and electromagnetic perturbations of Reissner-Nordstr{\"o}m spacetime for small charge, as has been proved in \cite{Giorgi6}.

The outline of the paper is as follows. 

In Section 2 we recall the null frame decomposition of the spacetime. In Section 3 we present the relevant structure of the Reissner-Nordstr{\"o}m spacetime and the linearized Einstein-Maxwell equations around it. 

In Section 4 we define the main quantities $\tilde{\b}$ and $\tilde{\bb}$ veryfing the Teukolsky-type equation of spin $1$, presented in Section 5. In Section 6, we present the Chandrasekhar transformation relating the Teukolsky-type equation to the Fackerell-Ipser-type equation. 

In Section 7 we define the main weighted energies and state the theorem. In Section 8 we derive the estimates for the projection of $\pf$ into the $\ell=1$ spherical mode. In Section 9, we derive the estimates for the solutions to the Teukolsky equation of spin $\pm1$, therefore proving the Main Theorem. 

In Appendix A we collect the computations in the derivation of the generalized Teukolsky equation of spin $\pm1$ and in Appendix B we show the derivation of the generalized Fackerell-Ipser equation through the Chandrasekhar transformation.

\section{Decomposition in null frames}
\label{VEeqDNGsec}
In this section, we review the formalism of local null frames of a Lorentzian manifold.

Let $(\MM, \g)$ be a $3+1$-dimensional Lorentzian manifold, and let $\D$ be the covariant derivative associated to $\g$.

Suppose that the the Lorentzian manifold $(\MM, \g)$ can be foliated by spacelike $2$-surfaces $(S,\slashed{g}) $, where $\slashed{g}$ is the pullback of the metric $\g$ to $S$.  To each point of $\MM$, we can associate a null frame $\mathscr{N}=\left\{e_A, e_3, e_4\right\}$, with $\{ e_A \}_{A=1,2}$ being tangent vectors to $(S,\slashed{g}) $, such that the following relations hold
\bea\label{relations-null-frame}
\begin{split}
\g\left(e_3,e_3\right) &= 0, \qquad \g\left(e_4,
e_4 \right) = 0, \qquad \g\left(e_3,e_4\right) = -2\\ 
\g\left(e_3,
e_A \right) &= 0 \ \ \ , \ \ \  \g\left(e_4, e_A\right) = 0 \ \ \ , \ \ \ \g \left(e_A, e_B \right) = \slashed{g}_{AB} \, .
\end{split}
\eea

 We define the Ricci coefficients associated to the metric $\g$ with respect to the null frame $\mathscr{N}$ in the following way (see \cite{Ch-Kl}):
   \bea\label{def1}
   \begin{split}
   \chi_{AB}:&=\g(\D_A  e_4, e_B), \qquad  \chib_{AB}:=\g(\D_A  e_3, e_B)\\
   \eta_A:&=\frac 1 2 \g(\D_3 e_4, e_A), \qquad \etab_A:=\frac 1 2 \g(\D_4 e_3, e_A),\\
    \xi_A:&=\frac 1 2 \g(\D_4 e_4, e_A), \qquad \xib_A:=\frac 1 2 \g(\D_3 e_3, e_A)\\
     \om:&=\frac 1 4 \g(\D_4 e_4, e_3), \qquad \omb:=\frac 1 4\g(\D_3 e_3, e_4)\\
   \ze_A:&=\frac 1 2 \g( \D_A e_4, e_3), \qquad 
   \end{split}
   \eea
   We decompose the $2$-tensor $\chi_{AB}$ into its tracefree part $\chih_{AB}$, a symmetric traceless 2-tensor on $S$, and its trace $\ka:=\tr\chi$. Similarly for $\chib_{AB}$.

Let $\W$ denote the Weyl curvature of $\g$ and let $\dual \W$ denote the Hodge dual on $(\MM, \g)$ of $\W$. We define the null curvature components in the following way (see \cite{Ch-Kl}):  
   \bea\label{def3}
\begin{split}
\a_{AB}:&=\W(e_A, e_4, e_B, e_4), \qquad \aa_{AB}:= \W(e_A, e_3, e_B, e_3) \\
\b_A:&=\frac 1 2 \W(e_A, e_4, e_3, e_4), \qquad  \bb_{A} :=\frac 1 2 \W(e_A, e_3, e_3, e_4) \\
\rho:&=\frac 14 \W(e_3, e_4, e_3, e_4) \qquad \sigma:=\frac 1 4  \dual \W (e_3, e_4, e_3, e_4)
\end{split}
\eea

Let $\F$ be a $2$-form in $(\MM, \g)$, and let $\dual \F$ denote the Hodge dual on $(\MM, \g)$ of $\F$. We define the null electromagnetic components in the following way:
\bea\label{decomposition-F}
\begin{split}
\bF_{A}:&=\F(e_A, e_4), \qquad \bbF_{A}:= \F(e_A, e_3) \\
 \rhoF:=& \frac 1 2 \F(e_3, e_4), \qquad \sigmaF:=\frac 1 2 \dual \F(e_3, e_4) 
\end{split}
\eea

If $(\MM, \g)$ satisfies the Einstein-Maxwell equations
\bea \label{Einstein-Maxwell-eq}
\R_{\mu\nu}&=&2 \F_{\mu \lambda} {\F_\nu}^{\lambda}- \frac 1 2 \g_{\mu\nu} \F^{\alpha\beta} \F_{\alpha\beta}, \label{Einstein-1}\\
\D_{[\alpha} \F_{\beta\gamma]}&=&0, \qquad \D^\alpha \F_{\alpha\beta}=0. \label{Maxwell}
\eea
the Ricci coefficients, curvature and electromagnetic components defined in \eqref{def1}, \eqref{def3} and \eqref{decomposition-F} satisfy a system of equations, presented in Section 2.3 of \cite{Giorgi4}.

\section{The Reissner-Nordstr{\"o}m spacetime}\label{RN-generale}

In this section, we introduce the Reissner-Nordstr{\"o}m exterior metric, as well as relevant background structure. We remind of Section 3 of \cite{Giorgi4} for a more complete description.

Define the manifold with boundary
\begin{align} \label{SchwSchmfld}
\mathcal{M} := \mathcal{D} \times S^2 := \left(-\infty,0\right] \times \left(0,\infty\right) \times S^2
\end{align}
with Kruskal coordinates $\left(U,V,\theta^1,\theta^2\right)$, as defined in Section 3 of \cite{Giorgi4}.
 The boundary $\mathcal{H}^+$ will be referred to as the \emph{horizon}. 
We denote by $S^2_{U,V}$ the $2$-sphere $\left\{U,V\right\} \times S^2 \subset \mathcal{M}$ in $\mathcal{M}$.

Fix two parameters $M>0$ and $Q$, verifying $|Q|<M$. Then the Reissner-Nordstr{\"o}m metric $g_{M, Q}$ with parameters $M$ and $Q$ is defined to be the metric:
\begin{align} \label{sskruskal}
g_{M, Q} = -4 \Omega_K^2 \left(U,V\right) d{U} d{V} +   r^2 \left(U,V\right) \gamma_{AB} d{\theta}^A d{\theta}^B.
\end{align}
where 
\beaa
\Omega_K^2 \left(U,V\right) &=& \frac{r_{-}r_{+}}{4r(U,V)^2} \Big( \frac{r(U,V)-r_{-}}{r_{-}}\Big)^{1+\left(\frac{r_{-}}{r_{+}}\right)^2}\exp\Big(-\frac{r_{+}-r_{-}}{r_{+}^2} r(U,V)\Big) \\
 \gamma_{AB} &=& \textrm{standard metric on $S^2$} \, .
\eeaa
and 
\bea\label{definiion-rpm}
r_{\pm}=M\pm \sqrt{M^2-Q^2}
\eea
and $r$ is an implicit function of the coordinates $U$ and $V$\footnote{See equation (38) in \cite{Giorgi4} for the implicit definition.}.

\subsection{Double null coordinates $u$, $v$}
\label{EFdnulldef}

We define another double null coordinate system that covers  the interior of $\mathcal{M}$, modulo the degeneration of the 
angular coordinates. This coordinate system, 
$\left(u,v,\theta^1, \theta^2\right)$, is called  \emph{double null coordinates} and are defined via the relations
\begin{align} \label{UuVv}
U = -\frac{2r_{+}^2}{r_{+}-r_{-}}\exp\left(-\frac{r_{+}-r_{-}}{4r_{+}^2} u\right) \ \ \ \textrm{and} \ \ \ \ V = \frac{2r_{+}^2}{r_{+}-r_{-}}\exp \left(\frac{r_{+}-r_{-}}{4r_{+}^2} v\right) \, .
\end{align}
Using (\ref{UuVv}), we obtain the Reissner-Nordstr{\"o}m metric on the interior of $\MM$ in $\left(u,v,\theta^1, \theta^2\right)$-coordinates:
\begin{align} \label{ssef}
g_{M, Q} =  - 4 \Omega^2 \left(u,v\right) \, d{u} \, d{v} +   r^2 \left(u,v\right) \gamma_{AB} d{\theta}^A d{\theta}^B  
\end{align}
with
\begin{align}
\label{officialOmegadef}
\Omega^2:= 1-\frac{2M}{r}+\frac{Q^2}{r^2}
\end{align}
 In $\left(u,v,\theta^1, \theta^2\right)$-coordinates, the horizon $\mathcal{H}^+$ can still be formally parametrised by $\left(\infty, v,\theta^2,\theta^2\right)$ with $v \in \mathbb{R}$, $\left(\theta^1,\theta^2\right) \in S^2$. We denote $r_{\mathcal{H}}=r_{+}=M+\sqrt{M^2-Q^2}$. We denote by $S_{u, v}$ the sphere $S^2_{U, V}$ where $U$ and $V$ are given by \eqref{UuVv}.

The photon sphere of Reissner-Nordstrom corresponds to the hypersurface in which null geodesics are trapped. It is the hypersurface given by $\{ r=r_P \}$ where $r_P$ is given by 
\bea\label{def-rP}
r_P=\frac{3M+\sqrt{9M^2-8Q^2}}{2}
\eea

\subsection{Null frames: Ricci coefficients and curvature components }\label{null-frames-ricci}

We define in this section two normalized null frames associated to Reissner-Nordstr{\"o}m. 
\begin{enumerate}
\item The symmetric null frame $(e_3, e_4)$ is given by 
\bea\label{symmetric-null-pair}
e_4 =\frac{1}{\Omega} \pr_v , \qquad e_3=\frac{1}{\Omega}\pr_u .
\eea 
In this null frame, 
\beaa
\ka=-\kab= \frac{2\Omega}{r}, \qquad \om=-\omb =-\frac{M}{2\Omega r^2} +\frac{Q^2}{2\Omega r^3}  
\eeaa
We also have 
\bea\label{nabb-3-Omega}
\nabb_3 \Omega&=& -2\omb \Omega, \qquad \nabb_4\Omega=-2\om \Omega
\eea
\item The regular\footnote{This frame extends regularly to a non-vanishing null frame on $\mathcal{H}^+$.} null frame $(e_3^*, e_4^*) $  is given by
\bea
\label{eq:regular-nullpair}
e_3^*=\Omega^{-1} e_3,\,\, \qquad e_4^*=\Omega e_4.
\eea
\end{enumerate}

The curvature and electromagnetic components which are non-vanishing do not depend on the particular null frame. They are given by 
\beaa
\rhoF=\frac{Q}{r^2}, \qquad \rho  =  -\frac{2M}{r^3}+\frac{2Q^2}{r^4} 
\eeaa
 We also have that
\begin{align} \label{GCdef}
K = \frac{1}{r^2}
\end{align}
for the Gauss curvature of the round $S^2$-spheres.

\subsection{Killing fields of the Reissner-Nordstr{\"o}m metric}\label{Killing-fields}

We discuss now the Killing fields associated to the metric $g_{M, Q}$.

We define the vectorfield $T$ to be the timelike Killing vector field $\pr_t$ of the $(t, r)$ coordinates in \eqref{RNintro}, which in double null coordinates is given by 
\beaa
T=\frac 1 2 (\pr_u+\pr_v)
\eeaa
 The vector field extends to a smooth Killing field on the horizon $\mathcal{H}^+$, which is moreover null and tangential to the null generator of $\mathcal{H}^+$.

We can also define a basis of angular momentum operator $\Omega_i$, $i=1,2,3$ (see for example Section 3.3 of \cite{Giorgi4}). The Lie algebra of Killing vector fields of $g_{M,Q}$ is then generated by $T$ and $\Omega_i$, for $i=1,2,3$.

\subsection{The spherical harmonics}\label{spherical-harmonics}

We collect some known definitions and properties of the Hodge decomposition of scalars, one forms and symmetric traceless two tensors in spherical harmonics. We also recall some known elliptic estimates. See Section 4.4 of \cite{DHR} for more details.

\subsubsection{The $\ell=0,1$ spherical harmonics and tensors supported on $\ell\ge 2$}
We denote by $\dot{Y}_m^\ell$, with $|m|\leq \ell$, the well-known spherical harmonics on the unit sphere, i.e. $$\lapp_0 \dot{Y}^\ell_m=- \ell(\ell+1) \dot{Y}^\ell_m$$ where $\lapp_0$ denotes the laplacian on the unit sphere $S^2$.  The $\ell=0,1$ spherical harmonics are given explicitly by 
\bea
\dot{Y}_{m=0}^{\ell=0}&=&\frac{1}{\sqrt{4\pi}}, \label{spherical-l=0}\\
\dot{Y}_{m=0}^{\ell=1}&=&\sqrt{\frac{3}{8\pi}} \cos \th, \qquad \dot{Y}_{m=-1}^{\ell=1}=\sqrt{\frac{3}{4\pi}} \sin\th \cos \phi, \qquad \dot{Y}_{m=1}^{\ell=1}=\sqrt{\frac{3}{4\pi}} \sin\th \sin \phi \label{spherical-l=1}
\eea
This family is orthogonal with respect to the standard inner product on the sphere, and any arbitrary function $f \in L^2(S^2)$ can be expanded uniquely with respect to such a basis. 

In the foliation of Reissner-Nordstr\"om spacetime, we are interested in using the spherical harmonics with respect to the sphere of radius $r$. For this reason, we normalize the definition of the spherical harmonics on the unit sphere above to the following.

We denote by $Y_m^\ell$, with $|m|\leq \ell$, the spherical harmonics on the sphere of radius $r$, i.e. $$\lapp Y^\ell_m=-\frac{1}{r^2} \ell(\ell+1) Y^l_m$$ where $\lapp$ denotes the laplacian on the sphere $S_{u,v}$ of radius $r=r(u, v)$. Such spherical harmonics are normalized to have $L^2$ norm in $S_{u, v}$ equal to $1$, so they will in particular be given by $Y_m^\ell=\frac{1}{r}\dot{Y}_m^\ell$. We use these basis to project functions on Reissner-Nordstr\"om manifold in the following way.  

\begin{definition}\label{lemma-spherical-harmonics} We say that  a function $f$ on $\mathcal{M}$ is supported on $\ell\geq 2$ if the projections 
\beaa
\int_{S_{u, v}}  f \c Y^\ell_m=0
\eeaa
vanish for $Y_m^{\ell=1}$ for $m=-1, 0, 1$. Any function $f$ can be uniquely decomposed orthogonally as 
\bea\label{decomposition-f-spherical-harmonics}
f=c(u, v) Y^{\ell=0}_{m=0}+\sum_{i=-1}^{1}c_i(u, v) Y^{\ell=1}_{m=i}(\th, \vphi)+f_{\ell\ge 2}
\eea
where $f_{\ell\ge 2}$ is supported in $\ell\ge2$.
\end{definition} 

In particular, we can write the orthogonal decomposition
\beaa
f=f_{\ell=0}+f_{\ell=1}+f_{\ell\geq2}
\eeaa
where 
\bea\label{explicit-formula-projection=l1}
f_{\ell=0}&=& \frac{1}{4\pi r^2} \int_{S_{u,v}}  f \\
f_{\ell=1}&=&\sum_{i=-1}^{1} \left( \int_{S_{u, v}} f \c Y^{\ell=1}_{m=i} \right) Y^{\ell=1}_{m=i} 
\eea

Recall that an arbitrary one-form $\xi$ on $S_{u,v}$ has a unique representation $\xi=r \DDs_1(f, g)$, for two uniquely defined functions $f$ and $g$ on the unit sphere, both with vanishing mean. In particular, the scalars $\divv \xi$ and $\curll \xi$ are supported in $l\ge1$. As in \cite{DHR}, we define 
\begin{definition}\label{decomposition-xi} We say that a smooth $S_{u,v}$ one form $\xi$ is supported on $\ell\ge 2$ if the functions $f$ and $g$ in the unique representation $$\xi=r \DDs_1(f, g)$$ are supported on $\ell\geq 2$. Any smooth one form $\xi$ can be uniquely decomposed orthogonally as $$\xi=\xi_{\ell=1}+\xi_{\ell\ge 2}$$ where the two scalar functions $r\DDd_1\xi=(r\divv\xi_{\ell=1}, r\curll\xi_{\ell=1})$ are in the span of \eqref{spherical-l=0} and $\xi_{\ell\ge 2}$ is supported on $\ell\ge 2$.
\end{definition}

Recall that an arbitrary symmetric traceless two-tensors $\th$ on $S_{u,v}$ has a unique representation $$\th=r^2\DDs_2\DDs_1(f, g)$$ for two uniquely defined functions $f$ and $g$ on the unit sphere, both supported in $\ell\ge 2$. In particular, the scalars $\divv \divv \th$ and $\curll \divv \th$ are supported in $\ell\ge 2$.

For future reference, we recall the following lemma.

\begin{lemma}[Lemma 4.4.1 in \cite{DHR}]\label{lemma-kernel-DDs2} The kernel of the operator $\mathcal{T}=r^2 \DDs_2 \DDs_1$ is finite dimensional. More precisely, if the pair of functions $(f_1, f_2)$ is in the kernel, then 
\beaa
f_1=c Y^{\ell=0}_{m=0}+\sum_{i=-1}^{1}c_i Y^{\ell=1}_{m=i}, \qquad f_2=\tilde{c} Y^{\ell=0}_{m=0}+\sum_{i=-1}^{1}\tilde{c}_i Y^{\ell=1}_{m=i}
\eeaa
for constants $c, c_i, \tilde{c}, \tilde{c}_i$.
\end{lemma}

\subsubsection{Elliptic estimates}\label{section-elliptic-estimates}

We recall the relations between the angular operators and the laplacian $\lapp$ on $S$:
\bea\label{angular-operators}
\begin{split}
\DDd_1 \DDs_1&=-\lapp_0, \qquad \DDs_1 \DDd_1= -\lapp_1+K, \\
\DDd_2 \DDs_2&= -\frac 1 2 \lapp_1-\frac 1 2 K, \qquad \DDs_2 \DDd_2= -\frac 1 2 \lapp_2+K
\end{split}
\eea
where $\lapp_0$ and $\lapp_1$ are the laplacian on scalars and on $1$-form respectively, and $K$ is the Gauss curvature of the surface $S$.

We recall the following $L^2$ elliptic estimates. See for example \cite{Ch-Kl}.

\begin{proposition}[Proposition 2.1.3 in \cite{stabilitySchwarzschild}] Let $(S, \gamma)$ be a compact surface with Gauss curvature $K$. Then the following identities hold for vectors $\xi$ on $S$:
\bea
\int_{S} \left(|\nabb \xi|^2+K |\xi|^2\right) &=& \int_{S} \left(|\divv \xi|^2+|\curll \xi|^2 \right)=\int_{S} |\DDd_1 \xi|^2 \label{first-elliptic-estimate}\\
\int_{S} \left(|\nabb \xi|^2-K |\xi|^2\right) &=& 2\int_{S} |\DDs_2 \xi|^2 
\eea
Moreover, suppose that the Gauss curvature is bounded. Then there exists a constant $C>0$ such that the following estimate holds for all vectors $\xi$ on $S$ orthogonal to the kernel of $\DDs_2$:
\bea
\int_{S} \frac{1}{r^2} |\xi|^2 \leq C \int_{S} |\DDs_2 \xi|^2\label{second-elliptic-estimate}
\eea
\end{proposition}

Consider a one-form $\xi$ on $\MM$ and its decomposition $\xi=\xi_{\ell=1}+\xi_{\ell\geq 2}$ as in Definition \ref{decomposition-xi}. Then we have the following elliptic estimate.

\begin{corollary}\label{main-elliptic-estimate} Let $\xi$ be a one-form on $\MM$. Then there exists a constant $C >0$ such that the following estimate holds:
\beaa
\int_{S} |\xi|^2 \leq C \left(\int_{S} |r\divv \xi_{l=1}|^2+|r\curll \xi_{l=1}|^2 + |r\DDs_2 \xi|^2 \right)
\eeaa
\end{corollary}
\begin{proof} Using the orthogonal decomposition of $\xi$, we have
\beaa
\int_{S} |\xi|^2&=& \int_{S} |\xi_{\ell=1}|^2+\int_{S} |\xi_{\ell\geq 2}|^2
\eeaa
Observe that, according to Lemma \ref{lemma-kernel-DDs2}, $\xi_{\ell\geq 2}$ is in the kernel of $\DDs_2$.  
Applying \eqref{first-elliptic-estimate} to $\xi_{\ell=1}$ and \eqref{second-elliptic-estimate} to $\xi_{\ell\geq 2}$ we obtain the desired estimate. 
\end{proof}

 We derive the transport equation for the projection to the $\ell=1$ spherical harmonics of a function $f$ on $\MM$.
 
 \begin{lemma}\label{commutation-projection-l1} Let $f$ be a scalar function on $\MM$. Then
 \beaa
 \nabb_4(f_{\ell=1})&=& (\nabb_4f)_{\ell=1} \qquad  \nabb_3(f_{\ell=1})= (\nabb_3f)_{\ell=1}
 \eeaa
 Consequently we have
 \bea\label{commute-wave}
 \Box_{g_{M, Q}} (f_{\ell=1})&=& ( \Box_{g_{M, Q}} f)_{\ell=1}
 \eea
 \end{lemma}
 \begin{proof} Applying $\nabb_4$ to the expression for the projection to the $l=1$ spherical harmonics given by \eqref{explicit-formula-projection=l1}, we obtain
 \beaa
 \nabb_4(f_{\ell=1})&=&\sum_{i=-1}^{1}  \nabb_4\Big(\left(\int_{S} f \c Y^{\ell=1}_{m=i}\right) Y^{\ell=1}_{m=i} \Big)
 \eeaa
 Recall that the normalized spherical harmonics are defined as $Y^{\ell=1}_m=\frac{1}{r} \dot{Y}^{\ell=1}_m$, where $\dot{Y}^{\ell=1}_m$ are given by \eqref{spherical-l=1}, and therefore $\nabb_4(\dot{Y}^{\ell=1}_m)=0$. This implies
 \beaa
 \nabb_4(Y^{\ell=1}_m)=\nabb_4\left(\frac{1}{r} \dot{Y}^{\ell=1}_m\right)=\nabb_4\left(\frac{1}{r}\right) \dot{Y}^{\ell=1}_m=-\frac{1}{2r}\ov{\ka}  \dot{Y}^{\ell=1}_m=-\frac{1}{2}\ov{\ka} Y^{\ell=1}_m
 \eeaa 
We therefore obtain
 \beaa
 \nabb_4(f_{\ell=1})&=&\sum_{i=-1}^{1}  \nabb_4\left(\int_{S} f \c Y^{\ell=1}_{m=i}\right) Y^{\ell=1}_{m=i} +\sum_{i=-1}^{1} \left(\int_{S} f \c Y^{\ell=1}_{m=i}\right) \nabb_4Y^{\ell=1}_{m=i} \\
 &=&\sum_{i=-1}^{1}  \left(\int_{S}\nabb_4( f \c Y^{\ell=1}_{m=i})+\ka  f \c Y^{\ell=1}_{m=i}\right) Y^{\ell=1}_{m=i} +\sum_{i=-1}^{1} \left(\int_{S} f \c Y^{\ell=1}_{m=i}\right) (-\frac{1}{2}\ov{\ka} Y^{\ell=1}_m) \\
 &=&\sum_{i=-1}^{1}  \left(\int_{S}\nabb_4( f) \c Y^{\ell=1}_{m=i}+f \c \nabb_4(Y^{\ell=1}_{m=i})+\frac 1 2 \ka  f \c Y^{\ell=1}_{m=i}\right) Y^{\ell=1}_{m=i}  \\
  &=&\sum_{i=-1}^{1}  \left(\int_{S}\nabb_4( f) \c Y^{\ell=1}_{m=i}\right) Y^{\ell=1}_{m=i} = (\nabb_4f)_{\ell=1}
 \eeaa
 as desired. Similarly for $\nabb_3 f$. 
 \end{proof}

\subsection{Linearized Einstein-Maxwell equations}\label{linearized-equations}

We collect here the equations for linearized gravitational and electromagnetic perturbation of Reissner-Nordstr{\"o}m metric. 
Recall that in Reissner-Nordstr{\"o}m metric the following Ricci coefficients, curvature and electromagnetic components vanish:
\beaa
&\chih, \quad \chibh, \quad \eta, \quad \etab, \quad \ze, \quad \xi, \quad \xib \\
&\a, \quad \b, \quad \sigma, \quad \bb, \quad \aa \\
&\bF, \quad  \sigmaF, \quad \bbF
\eeaa
In particular, in writing the linearization of the equations of Section \ref{VEeqDNGsec}, we will neglet the quadratic terms, i.e. product of terms above which vanish in Reissner-Nordstr{\"o}m background.  
See Section 4 of \cite{Giorgi4}.

\subsubsection{Linearised null structure equations}
\bea
\nabb_3 \chibh+\left( \kab +2\omb\right) \chibh&=&-2\DDs_2\xib -\aa, \label{nabb-3-chibh} \\
\nabb_4 \chih+\left(\ka+2\om\right) \ \chih&=&-2\DDs_2\xi -\a , \label{nabb-4-chih}\\
\nabb_3\chih+\left(\frac 12  \kab -2 \omb\right) \chih   &=&  -2 \slashed{\mathcal{D}}_2^\star \eta-\frac 1 2 \ka \chibh   \label{nabb-3-chih}\\
\nabb_4\chibh+\left(\frac 12  \ka-2\om\right) \,\chibh &=&  -2 \slashed{\mathcal{D}}_2^\star \etab -\frac 1 2 \kab \chih \label{nabb-4-chibh}, 
\eea
\bea
\nabb_3 \ze+ \left(\frac 1 2 \kab-2\omb \right)\ze&=&2 \DDs_1( \omb, 0)-\left(\frac 1 2 \kab+2\omb \right)\eta +\left(\frac 1 2 \ka +2\om \right)\xib - \bb  -\rhoF\bbF, \label{nabb-3-ze} \\
\nabb_4 \ze+ \left(\frac 1 2 \ka-2\om \right)\ze&=&-2 \DDs_1( \om, 0)+\left(\frac 1 2 \ka+2\om \right)\etab -\left(\frac 1 2 \kab +2\omb \right)\xi - \b  -\rhoF\bF, \label{nabb-4-ze} \\
\nabb_4\xib-\nabb_3 \etab&=&-\frac 1 2 \kab\left(\eta-\etab \right)+4\om  \xib -\bb-\rhoF \bbF , \label{nabb-4-xib}\\
\nabb_3\xi-\nabb_4 \eta&=&\frac 1 2 \ka\left(\eta-\etab \right)+4\omb  \xi +\b+\rhoF \bF , \label{nabb-3-xi}
\eea
\bea
\nabb_3 \kab +\frac 1 2 \kab^2+2\omb \ \kab&=&2\divv \xib, \label{nabb-3-kab}\\
\nabb_4 \ka +\frac 1 2 \ka^2+2\om \ \ka&=&2\divv \xi, \label{nabb-4-ka}\\
\nabb_3 \ka+\frac 1 2 \ka\kab-2\omb \ka&=& 2\divv \eta +2\rho, \label{nabb-3-ka}\\
\nabb_4 \kab+\frac 1 2 \ka\kab-2\om \kab&=& 2\divv \etab +2\rho, \label{nabb-4-kab}
\eea
\bea
\divv \chibh&=&-\frac 1 2 \kab \ze-\frac 1 2 \DDs_1( \kab, 0) +\bb-\rhoF\bbF, \label{Codazzi-chib} \\
\divv \chih&=&\frac 1 2 \ka \ze-\frac 1 2 \DDs_1(\ka, 0) -\b +\rhoF\bF \label{Codazzi-chi}
\eea
\bea
\nabb_4\omb+\nabb_3\om&=& 4\om\omb+\rho+\rhoF^2,\label{nabb-4-omb} \\
\curll \eta&=&\sigma, \label{curl-eta}\\
\curll \etab&=& -\sigma \label{curl-etab}
\eea
\bea
K&=&-\frac 1 4 \ka\kab-\rho+\rhoF^2 \label{Gauss}
\eea

\subsubsection{Linearised Maxwell equations}
\bea
\nabb_3 \bF+\left(\frac 1 2 \kab-2\omb\right) \bF&=& -\DDs_1(\rhoF, \sigmaF) +2\rhoF \eta, \label{nabb-3-bF}\\
\nabb_4 \bbF+\left(\frac 1 2 \ka-2\om\right) \bbF&=&\DDs_1(\rhoF, -\sigmaF)-2\rhoF \etab \label{nabb-4-bbF}
\eea
\bea
\nabb_3\rhoF+ \kab \rhoF&=& -\divv\bbF \label{nabb-3-rhoF}\\
\nabb_4\rhoF+ \ka \rhoF&=&\divv\bF\label{nabb-4-rhoF}
\eea
\bea
\nabb_3\sigmaF+\kab \ \sigmaF&=& \slashed{\curl}\bbF\label{nabb-3-sigmaF}\\
\nabb_4\sigmaF+\ka \ \sigmaF&=& \slashed{\curl}\bF\label{nabb-4-sigmaF}
\eea

\subsubsection{Linearised Bianchi identities}
\bea
 \nabb_3\a+\left(\frac 1 2 \kab-4\omb \right)\a&=&-2 \DDs_2\, \b-3\rho\chih -2\rhoF \ \left(\DDs_2\bF +\rhoF\chih \right), \label{nabb-3-a}\\
 \nabb_4\aa+\left(\frac 1 2 \ka-4\om\right) \aa&=&2 \DDs_2\, \bb-3\rho\chibh +2 \rhoF \ \left(\DDs_2\bbF -\rhoF\chibh \right) , \label{nabb-4-aa}
 \eea
 \bea
 \nabb_3 \b+\left(\kab -2\omb\right) \b &=&\DDs_1(-\rho, \sigma) +3\rho  \eta +\rhoF\left(-\DDs_1(\rhoF, \sigmaF)  - \ka\bbF-\frac 1 2 \kab \bF    \right) , \label{nabb-3-b}\\
 \nabb_4 \bb+\left(\ka-2\om\right) \bb &=&\DDs_1(\rho, \sigma) -3 \rho \etab +\rhoF\left(\DDs_1(\rhoF, -\sigmaF)  - \kab\bF-\frac 1 2 \ka \bbF    \right) , \label{nabb-4-bb}
 \eea
 \bea
 \nabb_3 \bb+ \left( 2\kab  +2\omb\right) \bb &=&-\divv\aa-3\rho \xib  +\rhoF\left(\nabb_3\bbF+2\omb \bbF+2\rhoF \ \xib\right), \label{nabb-3-bb}\\
 \nabb_4 \b+ \left(2\ka+2\om\right) \b &=&\divv\a +3\rho\xi +\rhoF\ \left(\nabb_4\bF+2\om \bF-2\rhoF \xi\right) \label{nabb-4-b}
 \eea
\bea
\nabb_3\rho+\frac 3 2 \kab \rho&=&-\kab\rhoF^2-\divv\bb-\rhoF \ \divv\bbF \label{nabb-3-rho}, \\
\nabb_4\rho+\frac 3 2 \ka \rho&=&-\ka\rhoF^2 +\divv\b+\rhoF \ \divv\bF, \label{nabb-4-rho}
\eea
\bea
\nabb_3 \sigma+\frac 3 2 \kab \sigma&=&-\slashed{\curl}\bb - \rhoF \ \slashed{\curl}\bbF, \label{nabb-3-sigma}\\
\nabb_4 \sigma+\frac 3 2 \ka \sigma&=&-\slashed{\curl}\b - \rhoF  \ \slashed{\curl}\bF \label{nabb-4-sigma}
\eea

\section{Gauge-invariant quantities $\tilde{\b}$ and $\tilde{\bb}$}\label{gauge-inv-section}

In this section, we define the tensors for which we prove decay in the Main Theorem of this paper.

 In order to identify the gauge-invariant quantities we consider null frame transformations, i.e. linear transformations which take null frames into null frames. They are given by 
\beaa
e_4'&=&e^a \left(e_4+f^A e_A \right), \\
e_3'&=& e^{-a} \left(e_3+\underline{f}^A e_A \right), \\
e_A'&=& {O_A}^B e_B+\frac 1 2 \underline{f}_A e_4+\frac 1 2 f_A e_3
\eeaa
where $a$ is a scalar function, $f$ and $\underline{f}$ are $S_{u,v}$-tensors and ${O_A}^B$ is an orthogonal transformation of $(S_{u,v}, \slashed{g})$, i.e. ${O_A}^C {O_B}^D\slashed{g}_{CD}=\slashed{g}_{AB}$. See Lemma 2.3.1 of \cite{stabilitySchwarzschild} for a classification of such transformations. When $a$, $f$ and $\underline{f}$ are small perturbations of $0$ and $O_{AB}$ is a small perturbation of the identity, we check that the above transformation transform a null frame into another null frame. Indeed, for example
\beaa
\g(e_3', e_4')&=& \g(e^{-a} \left(e_3+\underline{f}^A e_A \right), e^a \left(e_4+f^B e_B \right))=g(e_3, e_4)=-2 \\
\g(e_3', e_A')&=& \g(e^{-a} \left(e_3+\underline{f}^B e_B \right),  {O_A}^C e_C+\frac 1 2 \underline{f}_A e_4+\frac 1 2 f_A e_3)\\
&=& \g(e^{-a} e_3, \frac 1 2 \underline{f}_A e_4)+\g(e^{-a} \underline{f}^B e_B, O_{A}^C e_C)+\text{quadratic terms in $f$, $\underline{f}$}\\
&=& \text{quadratic terms}
\eeaa
These transformations modify the Ricci coefficients and the curvature components of the spacetime at the linear level. 
We summarize here some transformation rules for Ricci coefficients and curvature components under a general null transformation. See also Proposition 2.3.4 of \cite{stabilitySchwarzschild}. 

\begin{proposition}\label{prop:transformations1} Under a general null transformation, the linear transformations of curvature and electromagnetic components are the following: 
\bea
\b'&=&\b+\frac 3 2 \rho f \\
\rho'&=& \rho, \\
\bb'&=&\bb-\frac 3 2 \rho \fb 
\eea
 \bea
 \bF'&=&\bF +f \rhoF, \label{bfprime}\\
 \rhoF'&=&\rhoF,\\
 \bbF'&=&\bbF - \fb \rhoF \label{bbfprime}
 \eea
\end{proposition}
\begin{proof} We have
\beaa
\b'_A&=& \frac 1 2 \W(e_A', e_4', e_3', e_4')\\
&=& \frac 1 2 \W({O_A}^B e_B+\frac 1 2 \underline{f}_A e_4+\frac 1 2 f_A e_3, e^a \left(e_4+f^C e_C \right), e^{-a} \left(e_3+\underline{f}^E e_E \right), e^a \left(e_4+f^D e_D \right))\\
&=& \b_A +\frac 1 2 f^D \W(e_A, e_4, e_3, e_D)+\frac 1 4 f_A \W(e_3, e_4, e_3, e_4) + \text{quadratic terms}\\
&=& \b_A +\frac 3 2 f_A \rho + \text{quadratic terms}
\eeaa
We have
\beaa
 \bF'_A&=&  \F(e_A', e_4')=\F({O_A}^B e_B+\frac 1 2 \underline{f}_A e_4+\frac 1 2 f_A e_3, e^a \left(e_4+f^C e_C \right))\\
 &=& \bF_A+\frac 1 2 f_A \F(e_3, e_4)+\text{quadratic terms} \\
 &=& \bF_A+ f_A \rhoF
\eeaa
Similarly we see that $\rho'=\rho+\text{quadratic terms}$ and $\rhoF'=\rhoF+\text{quadratic terms}$. 
\end{proof}

We define a gauge-invariant quantity in the context of linear stability in the following way.
\begin{definition} We say that $\Psi$ is a (linear) gauge-invariant quantity if under a linear null frame transformation it is modified only quadratically, i.e. if $\Psi'=\Psi+ \text{quadratic terms}$. 
\end{definition}

The Teukolsky equations we shall consider are wave equations for gauge-invariant quantities for linear gravitational and electromagnetic perturbations of Reissner-Nordstr{\"o}m spacetime. 

From the study of Maxwell equations in Schwarzschild spacetime (see \cite{Blue} or \cite{Federico}), it is known that the extreme null components $\bF$ and $\bbF$ verify a decoupled wave equation, called the Teukolsky equation of spin $\pm1$. However, it is immediate from \eqref{bfprime} and \eqref{bbfprime} in Proposition \ref{prop:transformations1}, that the extreme electromagnetic components $\bF$ and $\bbF$ are not gauge-invariant if $\rhoF$ is not zero in the background. In particular, in the case of coupled gravitational and electromagnetic perturbations of Reissner-Nordstr{\"o}m spacetime, for which $\rhoF=\frac{Q}{r^2}$, the spin $\pm1$ Teukolsky equations verified by $\bF$ and $\bbF$ (derived in Proposition \ref{Teukolsky-bF}) cannot be used to obtain decay, since they are not gauge-invariant.

A new gauge-invariant quantity to be used in this context has to be found. This quantity has to be a one-tensor, in order to transport electromagnetic radiation, supported in $\ell\geq 1$.

We define the following $1$-tensors
\bea\label{definition-tilde-b}
\tilde{\b}:= 2\rhoF \b-3\rho \bF, \qquad \underline{\tilde{\b}}:= 2\rhoF \bb-3\rho \bbF
\eea
They are gauge invariant quantities. Indeed, using Proposition \ref{prop:transformations1}, we obtain
\beaa
\tilde{\b}'=2\rhoF' \b'-3\rho' \bF'=2\rhoF (\b+\frac 3 2 \rho f)-3\rho (\bF+\rhoF f)=2\rhoF \b-3\rho \bF=\tilde{\b}
\eeaa
and similarly for $\tilde{\bb}$. 

\begin{remark}\label{beta-tilde-schw} Observe that in Schwarzschild $\tilde{\b}$ reduces, at the linear level, to 
\beaa
\tilde{\b}:= -3\rho \bF=\frac{6M}{r^3} \bF
\eeaa
and the Teukolsky equation in Proposition \ref{Teukolsky-tilde-b} reduces to the Teukolsky equation in Proposition \ref{Teukolsky-bF} verified by $\bF$. 
\end{remark}

The gauge-invariant quantities $\tilde{\b}$ and $\tilde{\bb}$ have the additional remarkable property that they vanish for Kerr-Newman solutions, as it becomes clear in linearizing the Reissner-Nordstr{\"o}m spacetime for small variation of angular momentum (see Remark 5.2.1 in \cite{Giorgi6}).

This implies that, in order to prove linear stability of Reissner-Nordstr{\"o}m spacetime, since $\tilde{\b}$ and $\tilde{\bb}$ vanish for any pure gauge solutions and for any linearized Kerr-Newman solution, these quantities have to decay.

Indeed, $\tilde{\b}$ and $\tilde{\bb}$ verify a Teukolsky-type equation which can be treated through a Chandrasekhar transformation, and for which decay for the projection to the $\ell=1$ mode can be proved indipendently from the remaining projection. The boundedness and decay of the $\ell\geq 2$ modes are implied by the results for $\a$ and $\ff$ in \cite{Giorgi4}.

\section{Generalized spin $\pm1$-Teukolsky and Fackerell-Ipser equation in $l=1$ mode}\label{spin1}

In this section, we introduce a generalization of the celebrated spin $\pm 1$ Teukolsky equations and the Fackerell-Ipser equation, and explain the connection between them and their relation to the linear stability of Reissner-Nordstr{\"o}m spacetime to gravitational and electromagnetic perturbations.

\subsection{Generalized spin $\pm 1$ Teukolsky equation}

The generalized spin $\pm1$ Teukolsky equation concerns $1$-tensors which we denote $\tilde{\b}$ and $\tilde{\bb}$ respectively.

\begin{definition}\label{teukolsky-1} Let $\tilde{\b}$ be a $S^2_{u,v}$ 1-tensor defined
on a subset $\mathcal{D}\subset \mathcal{M}$. We say that $\tilde{\b}$ satisfy the {\bf generalized Teukolsky equation of spin ${\bf +1}$} if it satisfies the following PDE:
\beaa
\Box_\g (r^3 \tilde{\b})&=& -2\omb \nabb_4(r^3\tilde{\b})+\left(\ka+2\om\right) \nabb_3(r^3\tilde{\b})+\left(\frac {1}{4} \ka \kab-3\omb \ka+\om\kab-2\rho+3\rhoF^2-8\om\omb+2\nabb_3 \om\right) r^3 \tilde{\b}\\
&&-2r^3\kab\rhoF^2\left(\nabb_4 \bF+ (\frac 3 2 \ka+2\om)\bF-2\rhoF \xi \right)+\mathcal{I}
\eeaa
where $\Box_\g=\g^{\mu\nu} \D_{\mu}\D_{\nu}$ denotes the wave operator in Reissner-Nordstr{\"o}m spacetime, and $\mathcal{I}$ is a $1$-tensor with vanishing projection to the $\ell=1$ spherical harmonics, i.e. $\divv \ \mathcal{I}_{\ell=1}=\curll \ \mathcal{I}_{\ell=1}=0$.

Let $\tilde{\bb}$ be a $S^2_{u,v}$ 1-tensor defined
on a subset $\mathcal{D}\subset \mathcal{M}$. We say that $\tilde{\b}$ satisfy the {\bf generalized Teukolsky equation of spin ${\bf -1}$} if it satisfies the following PDE:
\beaa
\Box_\g (r^3 \tilde{\bb})&=& -2\om \nabb_3(r^3\tilde{\bb})+\left(\kab+2\omb\right) \nabb_4(r^3\tilde{\bb})+\left(\frac {1}{4} \ka \kab-3\om \kab+\omb\ka-2\rho+3\rhoF^2-8\om\omb+2\nabb_4 \omb\right) r^3 \tilde{\bb}\\
&&-2r^3\ka\rhoF^2\left(\nabb_3 \bbF+ (\frac 3 2 \kab+2\omb)\bbF+2\rhoF \xib \right)+\underline{\mathcal{I}}
\eeaa
where $\underline{\mathcal{I}}$ is a $1$-tensor with vanishing projection to the $\ell=1$ spherical harmonics.
\end{definition}

\subsection{Generalized Fackerell-Ipser equation in $l=1$ mode}

The other generalized equation in $\ell=1$ to be defined here is the generalized Fackerell-Ipser equation, to be satisfied by a one tensor $\pf$.

\begin{definition} \label{def:rwe-1}
Let $\pf$ be a $1$-tensor on $\mathcal{D}\subset \mathcal{M}$.
We say that $\pf$ satisfies the {\bf generalized Fackerell-Ipser equation in $\ell=1$} if it satisfies the following PDE:
\bea\label{finalequationl1}
\Box_\g\pf+\left(\frac 1 4 \ka\kab-5\rhoF^2\right)\pf &=&\mathcal{J}
\eea
where $\mathcal{J}$ is a $1$-tensor with vanishing projection to the $\ell=1$ spherical harmonics, i.e. $\divv \ \mathcal{J}_{\ell=1}= \curll \ \mathcal{J}_{\ell=1}=0$.
\end{definition}

In Section \ref{transformation-theory-1}, we will show that given a solution $\tilde{\b}$ and $\tilde{\bb}$ of the generalized spin $\pm 1$ Teukolsky equations in $\ell=1$, respectively, we can derive two solutions $\pf$ and $\underline{\pf}$, respectively, of the generalized Fackerell-Ipser equation in $\ell=1$.

\subsection{The characteristic initial value formulation}
For completeness, we state here a standard well-posedness theorem for both the generalized Teukolsky equation and the generalized Fackerell-Ipser equation. We formulate it in the context of a characteristic initial value problem. We fix a sphere $S_{u_0, v_0}$ in $\mathcal{M}$ and consider the outgoing Reissner-Nordstr{\"o}m light cone $C_{u_0}=\{u_0\} \times \{ v \geq v_0\} \times S^2$ and the ingoing Reissner-Nordstr{\"o}m light cone  $C_{v_0}=\{ u \geq u_0\}  \times \{ v_0\} \times S^2$ on which data are being prescribed.

\begin{proposition}[Well-posedness for generalized Teukolsky equation of spin $+1$] Given a sphere with corresponding null cones $C_{u_0}$ and $C_{v_0}$, prescribe 
\begin{itemize}
\item along $C_{v_0}$ a one $S_{u, v}$-tensor $\tilde{\b}_{0, in}$, such that $\Omega \tilde{\b}$ is smooth
\item along $C_{u_0}$ a smooth one $S_{u, v}$-tensor $\tilde{\b}_{0, out}$, satisfying $\tilde{\b}_{0, out}=\tilde{\b}_{0, in}$ on $S_{u_0, v_0}$.
\end{itemize}
Then there exists a unique smooth one $S_{u, v}$-tensor $\Omega \tilde{\b}$  defined on $\mathcal{M} \cap \{ u \geq u_0\} \cap \{ v \geq v_0\}$ such that 
\begin{itemize}
\item $\tilde{\b}$ satisfies the generalized Teukolsky equation of spin $+1$ in $\mathcal{M} \cap \{ u \geq u_0\} \cap \{ v \geq v_0\}$
\item $\Omega \tilde{\b}|_{C_{u_0}} = \Omega \tilde{\b}_{0, in}$ and $\tilde{\b}|_{C_{v_0}} = \tilde{\b}_{0, out}$
\end{itemize}
\end{proposition}

A similar result holds for the generalized Teukolsky equation of spin $-1$. 

The well-posedness statement for the Fackerell-Ipser equation is entirely analogous.

\section{Generalized Chandrasekhar transformation into Fackerell-Ipser}\label{transformation-theory-1}

We now describe a transformation theory relating solutions of the generalized
Teukolsky equation to solutions of the generalized Fackerell-Ipser equation in $\ell=1$. \footnote{Observe that the operators which allow to define the Chandrasekhar transformation into the Regge-Wheeler equations in \cite{Giorgi4} are the same as the operators for the Chandrasekhar transformation into Fackerell-Ipser equation used here. }

We introduce the following operators for a $n$-rank $S$-tensor $\Psi$, previously defined in \cite{Giorgi4}:
\bea \label{operators}
\underline{P}(\Psi)&=&\frac{1}{\kab} \nabb_3(r \Psi), \qquad  P(\Psi)=\frac{1}{\ka} \nabb_4(r \Psi) 
\eea 

Given a solution  $\tilde{\b}$ of the generalized Teukolsky equation of spin $+1$, we can define the following \emph{derived}\footnote{Recall the other derived quantities defined in \cite{Giorgi4}: \beaa
\psi_0 &=& r^2 \kab^2 \a, \\
\psi_1&=&\underline{P}(\psi_0), \\
\psi_2&=&\underline{P}(\psi_1)=\underline{P}(\underline{P}(\psi_0))=:\qf, \\
\psi_3&=& r^2 \kab \ \ff, \\
\psi_4&=& \underline{P}(\psi_3)=:\qf^\F
\eeaa} quantities for $\tilde{\b}$:
\bea\label{quantities-1}
\begin{split}
\psi_5 &= r^4 \kab \ \tilde{\b}, \\
\psi_6&=\underline{P}(\psi_5):=\pf
\end{split}
\eea
Similarly, given a solution  $\underline{\tilde{\b}}$ of the generalized Teukolsky equation of spin $-1$, we can define the following \emph{derived} quantities for $\underline{\tilde{\b}}$:
\bea\label{quantities-2}
\begin{split}
\underline{\psi}_5 &= r^4 \ka \ \tilde{\bb}, \\
\underline{\psi}_6&=P(\underline{\psi}_5):=\underline{\pf}
\end{split}
\eea

The following  proposition
is proven in Appendix \ref{derivation-fackerell}. 
\begin{proposition} \label{prop:rwt1-1}
Let $\tilde{\b}$ be a solution of the generalized Teukolsky equation of spin $+1$ on $\mathcal{M} \cap \{ u \geq u_0\} \cap \{ v \geq v_0\}$. Then the $1$-tensor $\pf$ as defined through \eqref{quantities-1} satisfies the generalized Fackerell-Ipser in $l=1$ system on $\mathcal{M} \cap \{ u \geq u_0\} \cap \{ v \geq v_0\}$. 

Similarly, let $\tilde{\bb}$ be a solution of the generalized Teukolsky equation of spin $-1$ on $\mathcal{M} \cap \{ u \geq u_0\} \cap \{ v \geq v_0\}$. Then the $1$-tensor $\underline{\pf}$ as defined through \eqref{quantities-2} satisfies the generalized Fackerell-Ipser equation in $l=1$ on $\mathcal{M} \cap \{ u \geq u_0\} \cap \{ v \geq v_0\}$. 
\end{proposition}
\begin{proof} In Lemma \ref{generalwavePhi}, we compute the wave equation verified by a derived quantity of the form $\underline{P}(\Psi)$. We use this lemma to derive the wave equation for $\pf$ in Proposition \ref{wave-eq-pf}, from the Teukolsky equation for $\tilde{\b}$. See Appendix \ref{derivation-fackerell}.
\end{proof}

The fact that the derived quantity $\pf$ satisfies the generalized Fackerell-Ipser equation in $\ell=1$, together with the transport relations \eqref{quantities-1}, will be the key to estimating the projection to $\ell=1$ spherical harmonics of $\tilde{\b}$ and control such projection of the electromagnetic contribution of the radiation.

\subsection{Relation with higher order quantities defined in Schwarzschild}\label{work-Federico}

As observed in Remark \ref{beta-tilde-schw}, the one-form $\tilde{\b}$ defined above for Reissner-Nordstr{\"o}m spacetime reduces to $\tilde{\b}=\frac{6M}{r^3} \bF$  in the particular case of Schwarzschild. Indeed, in the case of Maxwell equations in Schwarzschild spacetime, the one-form $\bF$ is gauge-invariant and verifies a Teukolsky equation of spin $\pm1$ given by 
\beaa
\square_\g \bF&=&-2\omb \nabb_4\bF +\left(\ka+2\om\right) \nabb_3(\bF)+\left(\frac 1 4 \ka \kab-3\omb\ka+\om\kab-2\nabb_4\omb\right) \bF +O(\ep^2)
\eeaa
as derived in Proposition \ref{Teukolsky-bF}. This is consistent with the fact that in Schwarzschild, with $\rhoF=O(\ep)$, the wave equation verified by $\tilde{\b}$ in Proposition \ref{Teukolsky-tilde-b} reduces to
\beaa
\Box_\g (r^3 \tilde{\b})&=& -2\omb \nabb_4(r^3\tilde{\b})+\left(\ka+2\om\right) \nabb_3(r^3\tilde{\b})+\left(\frac {1}{4} \ka \kab-3\omb \ka+\om\kab-2\nabb_4\omb\right) r^3 \tilde{\b}+O(\ep^2)
\eeaa
Observe that it coincides with the Teukolsky equation previously known for the Maxwell equations in Schwarzschild.

In the analysis of Maxwell equations in Schwarzschild in \cite{Federico}, the transformation theory is defined by the author in the following way\footnote{Observe that in \cite{Federico} the extreme null component is defined as $\a(V)=\F(V, L)=\F(V, \Omega e_4)=\Omega \bF$} (see Section 3.2 in \cite{Federico}):
\bea\label{phiFederico}
\phi:&=&\frac{r^2}{1-\mu} \nabb_{\underline{L}}(r\Omega\bF)
\eea
where $1-\mu=\Omega^2=1-\frac{2M}{r}$ and $\underline{L}=\Omega e_3$.
Recall that in double null coordinates used in \cite{Federico}, $\kab=\tr\chib=-\frac 2 r \Omega$. Substituting $\Omega=-\frac r 2 \kab$ in \eqref{phiFederico}, we obtain
\beaa
\phi&=&\frac{r^2}{\Omega} \nabb_{3}(r\Omega\bF)=\frac{r}{ \kab} \nabb_{3}(r^2  \kab\bF)=\frac{1}{6M}\frac{r}{ \kab} \nabb_{3}\left(r^5  \kab\tilde{\b}\right)=\frac{r}{6M} \pf
\eeaa
which relate the $\phi$ in \cite{Federico} in Schwarzschild to the $\pf$ defined in this paper.

\subsection{Relation with the linear stability of  Reissner-Nordstr{\"o}m spacetime}

We will now finally relate the equations presented above to the full system of linearized
gravitational and electromagnetic perturbations of Reissner-Nordstr{\"o}m spacetime in the context of linear stability of Reissner-Nordstr{\"o}m, as proved in \cite{Giorgi6}.

Consider a solution to the linearized Einstein-Maxwell equations  around Reissner-Nordstr{\"o}m spacetime, as presented in Section \ref{linearized-equations}.
Then, the quantities $\tilde{\b}_{A}=2\rhoF \b_A-3\rho \bF_A$ and $\tilde{\bb}_{A}=2\rhoF \bb_A-3\rho \bbF_A$ verify the generalized Teukolsky equation of spin $\pm1$ respectively. 
We obtain the following theorem.

\begin{theorem} \label{prop:relfull}
Let $\tilde{\b}$, $\tilde{\bb}$ be the curvature components of a solution to the linearized Einstein-Maxwell equations around Reissner-Nordstr{\"o}m spacetime as in Section \ref{linearized-equations}. Then $\tilde{\b}$ satisfies the generalized Teukolsky equation of spin $+1$, and $\tilde{\bb}$ satisfies the generalized Teukolsky equation of spin $-1$. Moreover, the derived quantities $\pf$ and $\underline{\pf}$ defined in \eqref{quantities-1} and \eqref{quantities-2} verify the generalized Fackerell-Ipser equation in $\ell=1$.
\end{theorem} 
\begin{proof} See Proposition \ref{Teukolsky-tilde-b} in Appendix A for the derivation of the generalized Teukolsky equation and Proposition \ref{wave-eq-pf} in Appendix B for the derivation of the generalized Fackerell-Ipser equation in $\ell=1$.

The quantity $\tilde{\b}$ verifies the following wave equation:
\beaa
\Box_\g (r^3 \tilde{\b})&=& -2\omb \nabb_4(r^3\tilde{\b})+\left(\ka+2\om\right) \nabb_3(r^3\tilde{\b})+\left(\frac {1}{4} \ka \kab-3\omb \ka+\om\kab-2\rho+3\rhoF^2-8\om\omb+2\nabb_3 \om\right) r^3 \tilde{\b}\\
&&-2r^3 \kab\rhoF^2 \left(\nabb_4 \bF+ (\frac 3 2 \ka+2\om)\bF-2\rhoF \xi \right)+8r^3 \rhoF^2 \divv \ff
\eeaa
which is of the form given in Definition \ref{teukolsky-1}, with $\mathcal{I}=8r^3 \rhoF^2 \divv \ff$, which indeed has vanishing projection to the $\ell=1$ mode. Indeed, recall that the tensor $\ff$ defined in \cite{Giorgi4} is a symmetric traceless two tensor, and therefore by results about spherical harmonics recalled in Section \ref{spherical-harmonics}, its divergence has vanishing projection to the $\ell=1$ spherical harmonics. By Proposition \ref{prop:rwt1-1}, the derived quantity $\pf$ verifies the following generalized Fackerell-Ipser equation in $\ell=1$:
\beaa
\Box_\g\pf+\left(\frac 1 4 \ka\kab-5\rhoF^2\right)\pf&=&8r^2 \rhoF^2\divv(\qf^\F)
\eeaa
which is of the form given in Definition \ref{def:rwe-1}, with $\mathcal{J}=8r^2 \rhoF^2\divv(\qf^\F)$, which has vanishing projection to the $\ell=1$ mode. 
 \end{proof}
 
 Using Proposition \ref{prop:rwt1-1}, we can associate to any solution to the linearized Einstein-Maxwell equations around Reissner-Nordstr{\"o}m spacetime a one form which verifies the generalized Fackerell-Ipser equation in $\ell=1$.

\subsection{The $\ell\geq 2$ modes of the electromagnetic radiation}

Recall the definition of the gauge-invariant quantities defined in \cite{Giorgi4}. From the Weyl curvature component $\alpha$ defined in \eqref{def3}, we defined the derived quantity $\psi_1$ as follows:
\beaa
\psi_1&=&\underline{P}(r^2 \kab^2 \a)= \frac 1 2  r (r^2 \kab^2\a)+\frac{1}{\kab} r\nabb_3(r^2 \kab^2\a)
\eeaa
We also defined the gauge-invariant quantity $\ff$ as 
\beaa
\ff &=&\DDs_2\bF+\rhoF \chih
\eeaa
and their equivalent spin $-2$ versions. These quantities defined in \cite{Giorgi4} are related to the one tensor $\tilde{\b}$ defined in \eqref{definition-tilde-b}. 
\begin{lemma}\label{relations-derived} The following relations hold true:
\beaa
r^3 \kab \DDs_2\tilde{\b}&=& -\rhoF\psi_1-\left(2\rhoF^2+3\rho \right)r^3\kab \ff, \\
r^3 \ka \DDs_2\tilde{\bb}&=& \rhoF\underline{\psi}_1-\left(2\rhoF^2+3\rho \right)r^3\ka \underline{\ff}
\eeaa
\end{lemma}
\begin{proof} By definition of $\ff$ we have 
\bea\label{ff-bF}
r^2\kab \ff &=&r^2 \kab \left(\DDs_2\bF+\rhoF \chih\right)
\eea
and by definition of $\psi_1$, we have
\beaa
\psi_1&=& \frac 1 2  r (r^2 \kab^2\a)+\frac{1}{\kab} r\left(r^2 \kab^2 \nabb_3(\a)-4 \omb r^2\kab^2 \a \right)=r^3 \kab \left(\nabb_3(\a)+\frac 1 2  \kab\a-4 \omb  \a \right)
\eeaa
Using Bianchi identity \eqref{nabb-3-a}, we obtain
\bea\label{psi-b}
\psi_1&=& r^3 \kab \left(-2\DDs_2\b-3\rho \chih-2\rhoF \ff \right)
\eea
Multiplying \eqref{ff-bF} by $3 r\rho$ and summing it to the \eqref{psi-b} multiplied by $\rhoF$, we obtain
\beaa
\rhoF\psi_1+\left(2\rhoF^2+3\rho \right)r^3\kab \ff&=&r^3 \kab \left(3\rho\DDs_2\bF-2\rhoF\DDs_2\b\right)=-r^3 \kab \DDs_2\tilde{\b}
\eeaa
as desired. Similarly for $\tilde{\bb}$. 
\end{proof}

By these relations, it is clear that the bounds and decay obtained for $\psi_1$ and $\ff$ in the Main Theorem in \cite{Giorgi4} imply bounds and decay for $\DDs_2 \tilde{\b}$, therefore on the projection to the $\ell\geq 2$ spherical harmonics of $\tilde{\b}$. Using the control for the $\ell=1$ mode of $\tilde{\b}$ and $\tilde{\bb}$ obtained through the generalized Fackerell-Ipser equation in $\ell=1$ and elliptic estimates, we can derive control for the one-tensors $\tilde{\b}$ and $\tilde{\bb}$.

\section{Energy quantities and statements of the Main Theorem}
We give the definitions of weighted energy quantities, and we provide the precise statement of the Main Theorem of this paper.

\subsection{Definition of weighted energies}
We define in this section a number of weighted energies.

We define the following vectorfields:
\begin{itemize}
\item $T=\frac 1 2 [\Omega \nabb_3+\Omega \nabb_4]$, 
\item $R^*=\frac 1 2 [-\Omega \nabb_3+\Omega \nabb_4]$
\end{itemize}
Notice that $T$ coincides with the Lie-differentiation $\mathcal{L}_T$ with respect to the Killing field $T$ as defined in Section \ref{Killing-fields}.

\subsubsection{Weighted energies for $\phi$}
The energies in this section will be applied to $\phi=(r^2\divv\pf)_{\ell=1}$,  $\phi=(r^2\curll\pf)_{\ell=1}$, or $\phi=(r^2\divv\underline{\pf})_{\ell=1}$,  $\phi=(r^2\curll\underline{\pf})_{\ell=1}$.

We introduce the following weighted energies for $\phi$. 

\begin{enumerate}
\item Flux energy quantities:
\begin{itemize}
\item $T$-energy null fluxes
\beaa
F_u^T[\phi](v_1, v_2)&=& \int_{v_1}^{v_2} dv \sin\th d\th d\phi \{|\Omega \nabb_4\phi|^2+\Omega^2|\nabb\phi|^2+V|\phi|^2 \}, \\
F_v^T[\phi](u_1, u_2)&=& \int_{u_1}^{u_2} du \sin\th d\th d\phi \{|\Omega \nabb_3\phi|^2+\Omega^2|\nabb\phi|^2+V|\phi|^2 \}
\eeaa
where $V=\frac{4Q^2}{r^4}\left( 1-\frac{2M}{r}+\frac{Q^2}{r^2}\right)$. \footnote{Hence these fluxes are manifestly coercive in the exterior region.}
\item Non-degenerate energy null fluxes
\beaa
F_u[\phi](v_1, v_2)&=& \int_{v_1}^{v_2} dv \sin\th d\th d\phi \{|\Omega \nabb_4\phi|^2+|\nabb\phi|^2+V|\phi|^2 \}, \\
F_v[\phi](u_1, u_2)&=& \int_{u_1}^{u_2} du \sin\th d\th d\phi \Omega^2 \{|\Omega^{-1} \nabb_3\phi|^2+|\nabb\phi|^2+V|\phi|^2 \}
\eeaa
\item Weighted energy quantity in null-infinity
\bea\label{definition-energy-far}
\begin{split}
F_{u}^{\mathcal{I}}[\phi](v_1, v_2)&= \int_{v_1}^{v_2} dv \sin\th d\th d\phi \Big[r^2 |\Omega \nabb_4\phi|^2+|\nabb \phi|^2+V|\phi|^2\Big], \\
F_{v}^{\mathcal{I}}[\phi](u_1, u_2)&= \int_{u_1}^{u_2} du \sin\th d\th d\phi \Omega^2 \Big[ |\Omega^{-1} \nabb_3\phi|^2+r^2|\nabb \phi|^2+|\phi|^2\Big]
\end{split}
\eea
\item Total energy flux:
\beaa
\mathbb{F}[\phi]&=& \sup_{u} F_{u}^\mathcal{I}[\phi](v_0, \infty)+\sup_{v} F^\mathcal{I}_{v}[\phi](u_0, \infty)
\eeaa
\item Initial energy at $(u_0, v_0)$:
\beaa
\mathbb{F}_0[\phi]&=& F_{u_0}^\mathcal{I}[\phi](v_0, \infty)+ F^\mathcal{I}_{v_0}[\phi](u_0, \infty)
\eeaa
\end{itemize}
\item Weighted spacetime bulk energies:
\begin{itemize}
\item Degenerate Morawetz bulk
\beaa
\mathbb{I}_{deg}[\phi]&=& \int_{u_0}^{\infty} \int_{v_0}^{\infty}  du dv\sin \th d\th d\phi \Omega^2 \Big[\frac{1}{r^2}|R^*\phi|^2+\frac{1}{r^3}|\phi|^2\\
&&+\frac{(r^2-3Mr+2Q^2)^2}{r^5}\left(|\nabb\phi|^2+\frac{1}{r^2}|\Omega \nabb_4 \phi|^2+\frac{1}{r^2}|\Omega^{-1} \nabb_3 \phi|^2 \right)\Big]
\eeaa
\item Weighted bulk norm in the far-away region
\beaa
\mathbb{I}_{\mathcal{I}}[\phi]&=& \int_{u_0}^{\infty} \int_{v_0}^{\infty}  du dv\sin \th d\th d\phi \c i_{r \geq R} \Big[r|\Omega\nabb_4 \phi|^2+r^{-1-\de} |\Omega \nabb_3 \phi|^2+r^{1-\de}|\nabb\phi|^2 +r^{-1-\de}|\phi|^2\Big]
\eeaa
where $ i_{r \geq R} $ is the indicator function which equals $1$ for $r \geq R$ and is zero otherwise. 
\end{itemize}
\end{enumerate}

\subsubsection{Weighted energies for $\tilde{\b}_{\ell=1}$ and $\tilde{\bb}_{\ell=1}$}
We define the following weighted energies for  $\tilde{\b}_{\ell=1}$ and $\tilde{\bb}_{\ell=1}$. 
\begin{enumerate}
\item Total energy fluxes:
\bea\label{definition-combined-esnergy}
\begin{split}
\mathbb{F}[\phi, \tilde{\b}_{\ell=1}]&= \mathbb{F}[\phi]+\sup_u \int_{v_0}^{\infty}  dv \sin \th d\th d\phi \Omega^{2} r^{10-\de} \left( |(r\divv\tilde{\b})_{\ell=1}|^2+   |(r\curll\tilde{\b})_{\ell=1}|^2\right) \\
\mathbb{F}[\phi, \tilde{\bb}_{\ell=1}]&= \mathbb{F}[\phi]+\sup_v \int_{u_0}^{\infty} du \sin \th d\th d\phi \ r^{8} \left(|(r\divv\tilde{\bb})_{\ell=1}|^2+ |(r\curll\tilde{\bb})_{\ell=1}|^2\right)
\end{split}
\eea
\item Total energy fluxes with first derivative:
\bea\label{definition-combined-esnergy-first-der}
\begin{split}
\mathbb{F}[\phi, \mathcal{D}\tilde{\b}_{\ell=1}]&= \mathbb{F}[\phi, \tilde{\b}_{l=1}]+\sup_u \int_{v_0}^{\infty}  dv \sin \th d\th d\phi \Omega^{2} r^{10-\de} \left( |\mathcal{D}(\Omega (r\divv\tilde{\b})_{\ell=1})|^2+   |\mathcal{D}(\Omega (r\curll\tilde{\b})_{\ell=1})|^2\right) \\
\mathbb{F}[\phi, \overline{\mathcal{D}}\tilde{\bb}_{\ell=1}]&= \mathbb{F}[\phi, \tilde{\bb}_{l=1}]+\sup_v \int_{u_0}^{\infty} du \sin \th d\th d\phi \ r^{8} \left(|\overline{\mathcal{D}}(\Omega^{-1} (r\divv\tilde{\bb})_{\ell=1})|^2+ |\overline{\mathcal{D}}(\Omega^{-1}(r\curll\tilde{\bb})_{\ell=1}|^2)\right)
\end{split}
\eea
where the notation $\mathcal{D}$ is defined\footnote{Observe that since those are energies for scalar supported in $l=1$ spherical harmonics, the control on the angular derivative is not needed, but is given by the zero-th order term.} by $|\mathcal{D}\xi|^2=|r\Omega^{-1} \nabb_3 \xi|^2+|r\Omega\nabb_4 \xi|^2$ and $\overline{\mathcal{D}}$ is $|\overline{\mathcal{D}}\xi|^2=|\Omega^{-1} \nabb_3 \xi|^2+|r\Omega\nabb_4 \xi|^2$.

\item Initial energy at $(u_0, v_0)$:
\beaa
\mathbb{F}_0[\phi, \tilde{\b}_{\ell=1}]&=& \mathbb{F}_0[\phi]+ \int_{v_0}^{\infty}  dv \sin \th d\th d\phi \Omega^{2} r^{10-\de} \left( |(r\divv\tilde{\b})_{\ell=1}|^2+   |(r\curll\tilde{\b})_{\ell=1}|^2\right)(u_0, v) \\
\mathbb{F}_0[\phi, \tilde{\bb}_{\ell=1}]&=& \mathbb{F}[\phi]+ \int_{u_0}^{\infty} du \sin \th d\th d\phi \ r^{8} \left(|(r\divv\tilde{\bb})_{\ell=1}|^2+ |(r\curll\tilde{\bb})_{\ell=1}|^2\right)(u, v_0)
\eeaa
and the equivalent definition for $\mathbb{F}_0[\phi, \mathcal{D}\tilde{\b}_{\ell=1}]$ and $\mathbb{F}_0[\phi, \mathcal{D}\tilde{\bb}_{\ell=1}]$. 
\item Bulk norms:
\bea\label{definition-combined-bulks}
\begin{split}
\mathbb{I}[\phi, \tilde{\b}_{\ell=1}]&= \mathbb{I}_{deg}[\phi]+\mathbb{I}_{\mathcal{I}}[\phi] +\int_{u_0}^{\infty} \int_{v_0}^{\infty} du dv \sin \th d\th d\phi \  \Omega^{4} r^{9-\de} \left(|(r\divv\tilde{\b})_{\ell=1}|^2+ |(r\curll\tilde{\b})_{\ell=1}|^2\right)\\
\mathbb{I}[\phi, \tilde{\bb}_{\ell=1}]&= \mathbb{I}_{deg}[\phi]+\mathbb{I}_{\mathcal{I}}[\phi] +\int_{u_0}^{\infty} \int_{v_0}^{\infty} du dv \sin \th d\th d\phi \   r^{7-\de} \left(|(r\divv\tilde{\bb})_{\ell=1}|^2+ |(r\curll\tilde{\bb})_{\ell=1}|^2\right)
\end{split}
\eea
\item Bulk norms with first derivatives:
\bea\label{definition-combined-bulks}
\begin{split}
\mathbb{I}[\phi, \mathcal{D}\tilde{\b}_{\ell=1}]&= \mathbb{I}_{deg}[\phi]+\mathbb{I}_{\mathcal{I}}[\phi] \\
&+\int_{u_0}^{\infty} \int_{v_0}^{\infty} du dv \sin \th d\th d\phi \  \Omega^{4} r^{9-\de} \left(|\mathcal{D}(\Omega (r\divv\tilde{\b})_{\ell=1})|^2+ |\mathcal{D}(\Omega (r\curll\tilde{\b})_{\ell=1})|^2\right)\\
\mathbb{I}[\phi, \overline{\mathcal{D}}\tilde{\bb}_{\ell=1}]&= \mathbb{I}_{deg}[\phi]+\mathbb{I}_{\mathcal{I}}[\phi] \\
&+\int_{u_0}^{\infty} \int_{v_0}^{\infty} du dv \sin \th d\th d\phi \   r^{7-\de} \left(|\overline{\mathcal{D}}(\Omega^{-1} (r\divv\tilde{\bb})_{\ell=1})|^2+ |\overline{\mathcal{D}}(\Omega^{-1}(r\curll\tilde{\bb})_{\ell=1})|^2\right)
\end{split}
\eea
\end{enumerate}

\subsubsection{Weighted energies for $\DDs_2\tilde{\b}$ and $\DDs_2\tilde{\bb}$}
We define the following weighted energies for  $\DDs_2\tilde{\b}$ and $\DDs_2\tilde{\bb}$. 
\begin{enumerate}
\item Total energy fluxes:
\bea\label{definition-combined-esnergy}
\begin{split}
\mathbb{F}[\DDs_2\tilde{\b}]&=\sup_u \int_{v_0}^{\infty}  dv \sin \th d\th d\phi \Omega^{2} r^{10-\de}  |r\DDs_2\tilde{\b}|^2 \\
\mathbb{F}[\DDs_2\tilde{\bb}]&=\sup_v \int_{u_0}^{\infty} du \sin \th d\th d\phi \ r^{8} |r\DDs_2\tilde{\bb}|^2
\end{split}
\eea
\item Total energy fluxes with first derivative:
\bea\label{definition-combined-esnergy-first-der}
\begin{split}
\mathbb{F}[\mathscr{D}\DDs_2\tilde{\b}]&= \mathbb{F}[\DDs_2\tilde{\b}]+\sup_u \int_{v_0}^{\infty}  dv \sin \th d\th d\phi \Omega^{2} r^{10-\de}  |\mathscr{D}(\Omega r\DDs_2\tilde{\b})|^2 \\
\mathbb{F}[\overline{\mathscr{D}}\DDs_2\tilde{\bb}]&= \mathbb{F}[\DDs_2\tilde{\bb}]+\sup_v \int_{u_0}^{\infty} du \sin \th d\th d\phi \ r^{8} |\overline{\mathscr{D}}(\Omega^{-1} r\DDs_2\tilde{\bb})|^2
\end{split}
\eea
where the notation $\mathscr{D}$ is defined by $|\mathscr{D}\xi|^2=|r\Omega^{-1} \nabb_3 \xi|^2+|r\Omega\nabb_4 \xi|^2+|r \nabb \xi|^2$ and $\overline{\mathscr{D}}$ is $|\overline{\mathscr{D}}\xi|^2=|\Omega^{-1} \nabb_3 \xi|^2+|r\Omega\nabb_4 \xi|^2+|r \nabb \xi|^2$.

We define in the obvious way the initial energy $\mathbb{F}_0[\DDs_2\tilde{\b}]$, $\mathbb{F}_0[\DDs_2\tilde{\bb}]$, $\mathbb{F}_0[\mathscr{D}\DDs_2\tilde{\b}]$ and $\mathbb{F}_0[\mathscr{D}\DDs_2\tilde{\b}]$.
\item Bulk norms:
\bea\label{definition-combined-bulks}
\begin{split}
\mathbb{I}[\DDs_2\tilde{\b}]&=\int_{u_0}^{\infty} \int_{v_0}^{\infty} du dv \sin \th d\th d\phi \  \Omega^{4} r^{9-\de} |r\DDs_2\tilde{\b}|^2\\
\mathbb{I}[\DDs_2\tilde{\bb}]&=\int_{u_0}^{\infty} \int_{v_0}^{\infty} du dv \sin \th d\th d\phi \   r^{7-\de} |r\DDs_2\tilde{\bb}|^2
\end{split}
\eea
\item Bulk norms with first derivatives:
\bea\label{definition-combined-bulks}
\begin{split}
\mathbb{I}[\mathscr{D}\DDs_2\tilde{\b}]&=\int_{u_0}^{\infty} \int_{v_0}^{\infty} du dv \sin \th d\th d\phi \  \Omega^{4} r^{9-\de} |\mathscr{D}(\Omega r\DDs_2\tilde{\b})|^2\\
\mathbb{I}[\overline{\mathscr{D}}\DDs_2\tilde{\bb}]&= \int_{u_0}^{\infty} \int_{v_0}^{\infty} du dv \sin \th d\th d\phi \   r^{7-\de}|\overline{\mathscr{D}}(\Omega^{-1} r\DDs_2\tilde{\bb})|^2
\end{split}
\eea
\end{enumerate}

\subsubsection{Weighted energies for $\tilde{\b}$ and $\tilde{\bb}$}
We define the following weighted energies for  $\tilde{\b}$ and $\tilde{\bb}$. 
\begin{enumerate}
\item Total energy fluxes:
\bea\label{definition-combined-esnergy}
\begin{split}
\mathbb{F}[\tilde{\b}]&=\sup_u \int_{v_0}^{\infty}  dv \sin \th d\th d\phi \Omega^{2} r^{10-\de}  |\tilde{\b}|^2 
+\Omega^{2} r^{10-\de}  |\mathscr{D}(\Omega \tilde{\b})|^2\\
\mathbb{F}[\tilde{\bb}]&=\sup_v \int_{u_0}^{\infty} du \sin \th d\th d\phi \ r^{8} |\tilde{\bb}|^2+r^{8} |\overline{\mathscr{D}}(\Omega^{-1} \tilde{\bb})|^2
\end{split}
\eea
and the initial energies $\mathbb{F}_0[\tilde{\b}]$ and $\mathbb{F}_0[\tilde{\bb}]$.
\item Bulk norms:
\bea\label{definition-combined-bulks}
\begin{split}
\mathbb{I}[\tilde{\b}]&=\int_{u_0}^{\infty} \int_{v_0}^{\infty} du dv \sin \th d\th d\phi \  \Omega^{4} r^{9-\de} |\tilde{\b}|^2+ \Omega^{4} r^{9-\de} |\mathscr{D}(\Omega \tilde{\b})|^2\\
\mathbb{I}[\tilde{\bb}]&=\int_{u_0}^{\infty} \int_{v_0}^{\infty} du dv \sin \th d\th d\phi \   r^{7-\de} |r\DDs_2\tilde{\bb}|^2+r^{7-\de}|\overline{\mathscr{D}}(\Omega^{-1} \tilde{\bb})|^2
\end{split}
\eea
\end{enumerate}

\subsubsection{Weighted energies defined in \cite{Giorgi4}}
We recall here the initial energy defined in \cite{Giorgi4}. This energy represents the initial data in the Main Theorem. 

The initial energy for $\qf$, $\qf^\F$,  $\underline{\qf}$ and $\underline{\qf}^\F$ are given by
\beaa
\mathbb{F}_0[\qf]&=& \int_{v_0}^{\infty} dv  r^2 \Big[r^2 |\Omega \nabb_4\qf|^2+|\nabb \qf|^2+r^{-2}|\qf|^2\Big] (u_0, v)\\
&& +\int_{u_0}^{\infty} du \Omega^2 r^2\Big[ |\Omega^{-1} \nabb_3\qf|^2+r^2|\nabb \qf|^2+|\qf|^2\Big] (u, v_0), \\
\mathbb{F}_0[\qf^\F]&=& \int_{v_0}^{\infty} dv  r^2\Big[r^2 |\Omega \nabb_4\qf^\F|^2+|\nabb \qf^\F|^2+r^{-2}|\qf^\F|^2\Big](u_0, v)\\
&&+\int_{u_0}^{\infty} du \Omega^2 r^2\Big[ |\Omega^{-1} \nabb_3\qf^\F|^2+r^2|\nabb \qf^\F|^2+|\qf^\F|^2\Big](u, v_0)
\eeaa
and identical energies for $\underline{\qf}$ and $\underline{\qf}^\F$.

The initial energy for $\ff$, $\psi_1$ and $\a$ are given by  
\beaa
\mathbb{F}_0[\ff]&=& \int_{v_0}^{\infty} dv r^2 \Big[r^6 |\Omega \nabb_4\ff|^2+|\nabb \ff|^2+r^{-2}|\ff|^2\Big](u_0, v)\\
&&+\int_{u_0}^{\infty} du \Omega^2 r^2\Big[ r^6|\Omega^{-1} \nabb_3\ff|^2+r^6|\nabb \ff|^2+r^4|\ff|^2\Big](u, v_0)\\
\mathbb{F}_0[\psi_1]&=& \int_{v_0}^{\infty} dv  r^2\Big[r^4 |\Omega \nabb_4\psi_1|^2+|\nabb \psi_1|^2+r^{-2}|\psi_1|^2\Big](u_0, v)\\
&&+\int_{u_0}^{\infty} du \Omega^2 r^2\Big[ r^4|\Omega^{-1} \nabb_3\psi_1|^2+r^4|\nabb \psi_1|^2+r^2|\psi_1|^2\Big](u, v_0)\\
\mathbb{F}_0[\a]&=& \int_{v_0}^{\infty} dv  r^2\Big[r^6 |\Omega \nabb_4\a|^2+|\nabb \a|^2+r^{-2}|\a|^2\Big](u_0, v)\\
&&+\int_{u_0}^{\infty} du \Omega^2 r^2\Big[ r^6|\Omega^{-1} \nabb_3\a|^2+r^6|\nabb \a|^2+r^4|\a|^2\Big](u, v_0)
\eeaa
The initial energy for $\underline{\ff}$, $\underline{\psi}_1$ and $\aa$ are given by  
\beaa
\mathbb{F}_0[\underline{\ff}]&=& \int_{v_0}^{\infty} dv r^2  \Big[r^4 |\Omega \nabb_4\underline{\ff}|^2+|\nabb \underline{\ff}|^2+r^{-2}|\underline{\ff}|^2\Big](u_0, v)\\
&&+\int_{u_0}^{\infty} du \Omega^2 r^2\Big[ r^2|\Omega^{-1} \nabb_3\underline{\ff}|^2+r^4|\nabb \underline{\ff}|^2+r^2|\underline{\ff}|^2\Big](u, v_0)\\
\mathbb{F}_0[\underline{\psi}_1]&=& \int_{v_0}^{\infty} dv r^2 \Big[r^2 |\Omega \nabb_4\underline{\psi}_1|^2+|\nabb \underline{\psi}_1|^2+r^{-2}|\underline{\psi}_1|^2\Big](u_0, v)\\
&&+\int_{u_0}^{\infty} du \Omega^2 r^2\Big[ |\Omega^{-1} \nabb_3\underline{\psi}_1|^2+r^2|\nabb \underline{\psi}_1|^2+|\underline{\psi}_1|^2\Big](u, v_0)\\
\mathbb{F}_0[\aa]&=& \int_{v_0}^{\infty} dv r^2 \Big[r^2 |\Omega \nabb_4\aa|^2+|\nabb \aa|^2+r^{-2}|\aa|^2\Big](u_0, v)\\
&&+\int_{u_0}^{\infty} du \Omega^2 r^2\Big[ |\Omega^{-1} \nabb_3\aa|^2+r^2|\nabb \aa|^2+|\aa|^2\Big](u, v_0)
\eeaa
We define
\beaa
\mathbb{F}_0[\qf, \qf^\F, \ff, \psi_1, \a]&=&\mathbb{F}_0[\qf]+\mathbb{F}^{1, T, \nabb}_0[\qf^\F]+\mathbb{F}^{1, T, \nabb}_0[\ff]+\mathbb{F}_0[\psi_1]+\mathbb{F}_0[\a] \\
\mathbb{F}_0[\underline{\qf}, \underline{\qf}^\F, \underline{\ff}, \underline{\psi}_1, \aa]&=&\mathbb{F}_0[\underline{\qf}]+\mathbb{F}^{1, T, \nabb}_0[\underline{\qf}^\F]+\mathbb{F}^{1, T, \nabb}_0[\underline{\ff}]+\mathbb{F}_0[\underline{\psi}_1]+\mathbb{F}_0[\aa]
\eeaa
The total initial energies are defined as 
\beaa
\mathbb{F}_0[\qf, \qf^\F, \ff, \psi_1, \a, \phi, \tilde{\b}]&=& \mathbb{F}_0[\qf, \qf^\F, \ff, \psi_1, \a]+\mathbb{F}_0[\phi, \tilde{\b}_{l=1}] +\mathbb{F}_0[\DDs_2\tilde{\b}]+\mathbb{F}_0[\mathscr{D}\DDs_2\tilde{\b}]\\
\mathbb{F}_0[\underline{\qf}, \underline{\qf}^\F, \underline{\ff}, \underline{\psi}_1, \aa, \underline{\phi}, \tilde{\bb}]&=& \mathbb{F}_0[\underline{\qf}, \underline{\qf}^\F, \underline{\ff}, \underline{\psi}_1, \aa]+\mathbb{F}_0[\phi, \tilde{\bb}_{l=1}] +\mathbb{F}_0[\DDs_2\tilde{\bb}]+\mathbb{F}_0[\mathscr{D}\DDs_2\tilde{\bb}]
\eeaa

\subsubsection{Higher order energies}
   
   To estimate higher order energies we also introduce the following notation, motivated by the fact that the Fackerell-Ipser equation commutes with $T$ and the angular momentum operators $\Omega_i$.  We define 
   \begin{enumerate}
\item Higher derivative energies for $n\geq 1$:
   \beaa
   \mathbb{F}^{n, T, \nabb}[\phi]&=& \sum_{i+j \leq n}\sup_{u} F_{u}^\mathcal{I}[T^i(r\nabb_A)^j\phi](v_0, \infty)+\sum_{i+j \leq n}\sup_{v} F^\mathcal{I}_{v}[T^i(r\nabb_A)^j\phi](u_0, \infty) \\
    \mathbb{F}_0^{n, T, \nabb}[\phi]&=& \sum_{i+j \leq n}F^\mathcal{I}_{u}[T^i(r\nabb_A)^j\phi](v_0, \infty)+\sum_{i+j \leq n}F^\mathcal{I}_{v}[T^i(r\nabb_A)^j\phi](u_0, \infty) 
   \eeaa
   and similarly for $\mathbb{F}^{n, T, \nabb}[\phi, \mathcal{D}\tilde{\b}_{\ell=1}]$,  $\mathbb{F}^{n, T, \nabb}[\phi, \overline{\mathcal{D}}\tilde{\bb}_{\ell=1}]$ and $\mathbb{F}^{n, T, \nabb}[\tilde{\b}]$. We also similarly define $\mathbb{F}^{n, T, \nabb}[\tilde{\bb}]$,  $\mathbb{F}^{n, T, \nabb}_0[\qf, \qf^\F, \ff, \psi_1, \a, \phi, \tilde{\b}]$.
     \item Higher derivative spacetime bulks:
   \bea
   \mathbb{I}_{deg}^{n, T, \nabb}[\phi]&=&\sum_{i+j \leq n} \mathbb{I}_{deg}[T^i(r\nabb_A)^j\phi], \qquad \mathbb{I}^{n, T, \nabb}_{\mathcal{I}}[\phi]= \sum_{i+j \leq n} \mathbb{I}_{\mathcal{I}}[T^i(r\nabb_A)^j\phi] \label{definition-higher-order}
   \eea
   and similarly for $\mathbb{I}^{n, T, \nabb}[\phi, \mathcal{D}\tilde{\b}_{\ell=1}]$, $\mathbb{I}^{n, T, \nabb}[\phi, \overline{\mathcal{D}}\tilde{\bb}_{\ell=1}]$ and $\mathbb{I}^{n, T, \nabb}[\tilde{\b}]$, $\mathbb{I}^{n, T, \nabb}[\tilde{\bb}]$
  \end{enumerate}

\subsection{Precise statement of the Main Theorem}\label{section-main-theorem}

We are now ready to state the boundedness and decay theorem for solutions $\tilde{\b}$ and $\tilde{\bb}$ of the generalized Teukolsky equation of spin $\pm 1$. 

In the estimates below we denote $\mathcal{A} \les \mathcal{B}$ if there exists an universal constant $C$ such that $\mathcal{A} \leq C \mathcal{B}$.

\begin{namedtheorem}[Main]\label{main-theorem} Let $\tilde{\b}$ and $\tilde{\bb}$ be solutions to the generalized Teukolsky equation of spin $\pm1$ respectively in Reissner-Nordstr{\"o}m spacetime for $|Q|\ll M$. Then the following estimates hold:
\begin{enumerate}
\item weighted boundedness and integrated decay estimate for $\tilde{\b}$ and $\tilde{\bb}$:
\bea\label{estimate-l=1-tilde-b-total}
\begin{split}
\mathbb{F}[\tilde{\b}]+\mathbb{I}[\tilde{\b}]&\les \mathbb{F}_0[\qf, \qf^\F, \ff, \psi_1, \a, \phi, \tilde{\b}]\\
\mathbb{F}[\tilde{\bb}]+\mathbb{I}[\tilde{\bb}]&\les \mathbb{F}_0[\underline{\qf}, \underline{\qf}^\F, \underline{\ff}, \underline{\psi}_1, \aa, \underline{\phi}, \tilde{\bb}]
\end{split}
\eea
\item higher order energy and integrated decay estimates for any integer $n\geq 1$:
\bea\label{estimate-higher-order-l=1-tilde-b-total}
\begin{split}
\mathbb{F}^{n, T, \nabb}[\tilde{\b}]+\mathbb{I}^{n, T, \nabb}[\tilde{\b}]&\les \mathbb{F}^{n, T, \nabb}_0[\qf, \qf^\F, \ff, \psi_1, \a, \phi, \tilde{\b}]\\
\mathbb{F}^{n, T, \nabb}[\tilde{\bb}]+\mathbb{I}^{n, T, \nabb}[\tilde{\bb}]&\les \mathbb{F}^{n, T, \nabb}_0[\underline{\qf}, \underline{\qf}^\F, \underline{\ff}, \underline{\psi}_1, \aa, \underline{\phi}, \tilde{\bb}]
\end{split}
\eea
\item pointwise decay estimates:
\bea\label{pointwise-decay}
|r^{\frac{9+\de}{2}} \tilde{\b}| \leq C v^{-\frac{2-\de}{2}}, \qquad |r^{4} \tilde{\bb}| \leq C v^{-\frac{2-\de}{2}} 
\eea
where $C$ depends on appropriate higher Sobolev norms. 
\end{enumerate}

\end{namedtheorem}

\subsection{The logic of the proof}

The remainder of the paper concerns the proof of the Main Theorem. We outline here the main steps of the proof.

\begin{enumerate}
\item Given a solution $\tilde{\b}$ and $\tilde{\bb}$ of the generalized Teukolsky equation of spin $\pm1$ we can associate a solution $\pf$ to the generalized Fackerell-Ipser equation in $\ell=1$ \eqref{finalequationl1} through Proposition \ref{prop:rwt1-1}. We commute equation \eqref{finalequationl1} with $r \DDd_1$ and project it into the  $\ell=1$ mode.  By Corollary \ref{wave-scalar-pf} and \eqref{commute-wave}, we have that the scalar quantities $(r\divv \pf)_{\ell=1}$ and $(r\curll \pf)_{\ell=1}$ verify the wave equation
\bea\label{equations-divv-curll}
\begin{split}
\Box_\g\left(r\divv\pf\right)_{\ell=1}-\left(\rho+4\rhoF^2\right)\left(r\divv\pf\right)_{\ell=1}&=0 \\
\Box_\g\left(r\curll\pf\right)_{\ell=1}-\left(\rho+4\rhoF^2\right)\left(r\curll\pf\right)_{\ell=1}&=0
\end{split}
\eea
and similarly for $(r\divv \underline{\pf})_{\ell=1}$ and $(r\curll \underline{\pf})_{\ell=1}$. Equations \eqref{equations-divv-curll} are scalar wave equations for which integrated local energy estimates can be obtained in a standard way. We obtain energy conservation, degenerate Morawetz estimates, redshift estimates and $r^p$ hierarchy estimates for the scalar quantity $(r\DDd \ \pf)_{\ell=1}$. This is done in Section \ref{mode-l=1-p}.
\item Using the relations \eqref{quantities-1}, and applying enhanced transport estimates, from the control obtained for $(r\DDd \ \pf)_{\ell=1}$ we will get estimates for the scalar quantities $(r\DDd \ \tilde{\b})_{\ell=1}$ and $(r\DDd \ \tilde{\bb})_{\ell=1}$, which give the $\ell=1$ spherical mode of $\tilde{\b}$ and $\tilde{\bb}$. This is done in Section \ref{mode-l=1-b}.
\item To control the projection to the $\ell\geq 2$ spherical modes of $\tilde{\b}$ and $\tilde{\bb}$ we will make use of the relations between $\tilde{\b}$ and $\a$ and $\ff$, and similarly for the negative spin quantities, shown in Lemma \ref{relations-derived}. We will then make use of the estimates obtained in \cite{Giorgi4} for $\a$ and $\ff$.  This is done in Section \ref{lge2b}.
\end{enumerate}

The two points above then imply the Main Theorem, i.e. the integrated decay statements about the solutions $\tilde{\b}$ and $\tilde{\bb}$ of the generalized Teukolsky equation of spin $\pm1$,  through standard elliptic estimates. This is finally done in Section \ref{proof-main-theorem}.

\section{Estimates for the $\ell=1$ spherical mode of $\pf$}\label{mode-l=1-p}

The present section contains the proof of the estimates for the projection to the $\ell=1$ spherical mode of the solution to the generalized Fackerell-Ipser equation $\pf$. This proof follows closely previous works for the scalar wave equation (see \cite{lectures}, \cite{redshift}, \cite{rp}), for the Regge-Wheeler equation (see \cite{DHR}) and for the Fackerell-Ipser equation (see \cite{Federico}).

We summarize the main result in the following Proposition.

\begin{proposition}\label{estimate-l=1-pf-proposition} Let $\pf$ and $\underline{\pf}$ be solutions to the generalized Fackerell-Ipser equation in $l=1$ \eqref{finalequationl1} in Reissner-Nordstr{\"o}m spacetime for $|Q|\ll M$. Let $\phi=(r^2\divv\pf)_{\ell=1}$ or  $\phi=(r^2\curll\pf)_{\ell=1}$, or $\phi=(r^2\divv\underline{\pf})_{\ell=1}$ or  $\phi=(r^2\curll\underline{\pf})_{\ell=1}$. Then the following estimates hold for all $\de \leq p \leq 2$:
\begin{enumerate} 
\item energy boundedness, degenerate integrated local energy decay and $r^p$ hierarchy estimates for $\phi$:
\bea\label{total-estimate-l=1}
\mathbb{F}[\phi]+ \mathbb{I}_{deg}[\phi]+\mathbb{I}_{\mathcal{I}}[\phi] \les \mathbb{F}_0[\phi]
\eea
\item higher order energy and integrated decay estimates for any integer $n\geq 1$:
\bea\label{higher-order-l=1}
\mathbb{F}^{n, T, \nabb}[\phi]+ \mathbb{I}_{deg}^{n, T, \nabb}[\phi]+\mathbb{I}^{n, T, \nabb}_{p, R}[\phi] \les \mathbb{F}_0^{n, T, \nabb}[\phi]
\eea
\end{enumerate}
\end{proposition}

We begin in Section \ref{eq-double-null} to write the projection to the $\ell=1$ mode of \eqref{finalequationl1} in double null coordinates. We derive the energy conservation for the obtained scalar wave equation in Section \ref{energy-conservation}. We then show a version of integrated decay which degenerates at the photon sphere, at the horizon and at null infinity in Section \ref{Morawetz}. We remove the degeneration at the horizon making use of the redshift vectorfield in Section \ref{redshift-est}, and we refine the degeneration at null infinity through the $r^p$ hierarchy estimates in Section \ref{rp-est}. Finally, we derive higher order energy estimates in Section \ref{higher-order-est}.

\subsection{The projection of the Fackerell-Ipser equation to the $\ell=1$ mode}\label{eq-double-null}

Let $\pf$ and $\underline{\pf}$ be a solution of the generalized Fackerell-Ipser equation in $\ell=1$ \eqref{finalequationl1}.

By \eqref{equations-divv-curll}, defining $\tilde{\phi}$ as $\tilde{\phi}=(r\divv\pf)_{\ell=1}$ or  $\tilde{\phi}=(r\curll\pf)_{\ell=1}$, or $\tilde{\phi}=(r\divv\underline{\pf})_{\ell=1}$,  $\tilde{\phi}=(r\curll\underline{\pf})_{\ell=1}$ then $\tilde{\phi}$ verifies 
\bea\label{box-tilde-phi}
\Box_\g\tilde{\phi}-\left(\rho+4\rhoF^2\right)\tilde{\phi}&=0
\eea
By Corollary \ref{connection-Federico}, equation \eqref{box-tilde-phi} can be written in double null coordinates as 
\beaa
 \Omega\nabb_3(\Omega \nabb_4(r \tilde{\phi}))-\Omega^2\lapp(r \tilde{\phi})+4\rhoF^2 \Omega^2 r \tilde{\phi}=0
\eeaa
Define $\phi=r\tilde{\phi}$ the equation becomes 
\bea\label{fackerell}
\Omega\nabb_3(\Omega \nabb_4(\phi))-\left( 1-\frac{2M}{r}+\frac{Q^2}{r^2}\right)\lapp \phi+V\phi=0
\eea
where $V=4\rhoF^2 \Omega^2=\frac{4Q^2}{r^4}\left( 1-\frac{2M}{r}+\frac{Q^2}{r^2}\right)$. 

\begin{remark} Equation \eqref{fackerell} reduces to the scalar form of the Fackerell-Ipser equation in Schwarzschild $(1-\mu)^{-1}L \underline{L} u -\lapp u =0$ when $Q=0$ (see equation 2.14 in \cite{Federico}). 
\end{remark}

\subsection{Energy conservation}\label{energy-conservation}
Let $\phi$ be a smooth solution of \eqref{fackerell} on $\{ u \geq u_0 \} \cap \{ v \geq v_0 \}$. By standard computations, the equation \eqref{fackerell} implies 
\beaa
&&[\Omega \nabb_3 +\Omega \nabb_4] \int \sin \th d\th d\phi \{ |\Omega \nabb_4 \phi|^2 +|\Omega \nabb_3 \phi|^2 +2 \frac{1-\frac{2M}{r}+\frac{Q^2}{r^2}}{r^2}|r\nabb\phi|^2+2V |\phi|^2 \}+\\
&&+[\Omega \nabb_3 -\Omega \nabb_4] \int \sin \th d\th d\phi \{ |\Omega \nabb_4 \phi|^2 -|\Omega \nabb_3 \phi|^2\}=0
\eeaa
Integrating the above with respect to $du dv$ yields a conservation law: for any $u \geq u_0$ and $v \geq v_0$ the $\phi$ of Proposition \ref{estimate-l=1-pf-proposition} satisfies
\bea\label{degenerate-energy}
F_u^T[\psi](v_0, v)+F_v^T[\psi](u_0, u)=F_{u_0}^T[\psi](v_0, v)+F_{v_0}^T[\psi](u_0, u)
\eea

\subsection{Morawetz estimates}\label{Morawetz}
Let $f$ be a function on $\mathcal{M}$ of $r^*=v-u$ only, and denote $f':=R^*(f)$. Following \cite{DHR}, equation \eqref{fackerell} implies the identity
\bea\label{integration-identity-morawetz}
\begin{split}
0&=_{S} [\Omega \nabb_3 +\Omega \nabb_4] \left( f(r)\{ |\Omega \nabb_4 \phi|^2 -|\Omega \nabb_3 \phi|^2 +f'(r) \phi [\Omega \nabb_3 +\Omega \nabb_4]\phi \}\right)\\
&+ [\Omega \nabb_3 -\Omega \nabb_4] \Big( f(r)\{ |\Omega \nabb_4 \phi|^2 +|\Omega \nabb_3 \phi|^2-\frac{2}{r^2} \frac{1-\frac{2M}{r}+\frac{Q^2}{r^2}}{r^2}|r\nabb\phi|^2-2V |\phi|^2\}\\
 &-f'(r) \phi [\Omega \nabb_3 -\Omega \nabb_4]\phi -f''(r)|\psi|^2\}\Big)\\
 &+2f'(r) |\Omega \nabb_4\psi -\Omega \nabb_3\psi|^2+|r\nabb\phi|^2 \left(-4f\left(\frac{\Omega^2}{r^2} \right)' \right)+|\phi|^2\left(-4fV'-2f''' \right)
 \end{split}
\eea
where the symbol $=_{S}$ means that the above identity holds adter integration over $\int \sin\th d\th d\phi$.

After integrating the above identity in the spacetime with respect to the measure $\int du dv \sin \th d\th d\phi$, the last line gives a spacetime energy, while the other lines give boundary terms. 

We want to choose $f(r)$ such that the last line in \eqref{integration-identity-morawetz} gives a coercive spacetime energy. The choice below is a generalization of the choice appeared in \cite{DHR} and \cite{Federico} to the case of the subextremal Reissner-Nordstr{\"o}m spacetime for which $|Q| \ll M$.

Recall that $\phi$ is a scalar function supported in $\ell=1$ spherical harmonics. Therefore we have the Poincar{\'e} inequality 
\beaa
\int_{S^2} |\nabb\phi|^2 \geq \frac{2}{r^2} \int_{S^2} |\phi|^2
\eeaa
Therefore, the last line of \eqref{integration-identity-morawetz} integrated on the sphere can be bounded from below by 
\bea\label{bulk-est}
\begin{split}
&&\int \sin \th d\th d\phi \{f'(r) |R^*\psi|^2+|r\nabb\phi|^2 \left(-4f\left(\frac{\Omega^2}{r^2} \right)' \right)+|\phi|^2\left(-4fV'-2f''' \right)\}\\
&\geq& \int \sin \th d\th d\phi \{f'(r) |R^*\psi|^2+\big( -4f\left(V+\frac{2\Omega^2}{r^2} \right)'-2f'''\big)|\phi|^2\}
\end{split}
\eea
We want to choose $f(r)$ such that 
\begin{enumerate}
\item $\frac{f'}{\Omega^2} \geq \frac{c}{r^2}$, 
\item  $-\frac{4f}{\Omega^2}\left(V+\frac{2\Omega^2}{r^2} \right)'-\frac{2f'''}{\Omega^2} \geq \frac{c}{r^3}$
\end{enumerate}
for some positive constants $c$. Indeed, the above conditions will give a coercive estimate for the bulk term in the Morawetz estimates.

We define 
\beaa
f&=& \left(1-\frac{3M}{r}+\frac{2Q^2}{r^2} \right)\left(1+\frac{M}{r}+\frac{Q^2}{2r^2} \right)
\eeaa
We denote with the subscript $r$ the derivative with respect to $r$. We compute 
\beaa
f_r&=& \frac{1}{2r^5} \left(4M r^3+(12M^2-10Q^2) r^2-3MQ^2 r-8Q^4 \right), \\
f_{rr}&=& \frac{1}{r^6} \left(-4Mr^3+(-18M^2+15Q^2)r^2+6MQ^2 r+20Q^4 \right), \\
f_{rrr}&=& \frac{1}{r^7} \left(12Mr^3+(72M^2-60Q^2)r^2-30MQ^2 r-120Q^4 \right)
\eeaa
In Reissner-Nordstr{\"o}m metric, define $0<\gamma<1$ such that $Q^2= \gamma M^2$. Then we obtain
\beaa
f_r&=& \frac{1}{2r^5} \left(4M r^3+(12-10\gamma)M^2 r^2-3\gamma M^3 r-8\gamma^2M^4 \right)
\eeaa
The horizon is then given by $r_{\mathcal{H}^+}=M+\sqrt{M^2-Q^2} = M( 1+\sqrt{1-\gamma})$. Define $r=M( 1+\sqrt{1-\gamma})x$ we have that 
\beaa
f_r&=& \frac{M^4}{2r^5} \left(4( 1+\sqrt{1-\gamma})^3 x^3+(12-10\gamma)( 1+\sqrt{1-\gamma})^2 x^2-3\gamma ( 1+\sqrt{1-\gamma}) x-8\gamma^2 \right)
\eeaa
 The polynomial on the right hand side is clearly positive in $x \in [1, \infty]$ for $\gamma=0$. Therefore, it is also positive for $\gamma$ small enough.

Now we analyze the expression $-\frac{4f}{\Omega^2}\left(V+\frac{2\Omega^2}{r^2} \right)'-\frac{2f'''}{\Omega^2}$. Recall that $V=\frac{4Q^2}{r^4}\Omega^2$, therefore
\bea\label{expression-to-be-positive}
-\frac{4f}{\Omega^2}\left(V+\frac{2\Omega^2}{r^2} \right)'-\frac{2f'''}{\Omega^2}&=& -\frac{4f}{\Omega^2}\left(\left(\frac{2Q^2}{r^2}+1\right)\left(\frac{2\Omega^2}{r^2}\right) \right)'-\frac{2f'''}{\Omega^2}
\eea
We compute the derivative of $f$ with respect to $R^*$, where recall that $R^*(r)=\Omega^2$. We therefore have
\beaa
f''(r)&=& \Omega^4 f_{rr}(r)+(\Omega^2)' \Omega^2 f_r(r), \\
f'''(r)&=& 3 (\Omega^2)'\Omega^4 f_{rr}+\Omega^6 f_{rrr}+(\Omega^2 ((\Omega^2)')^2+\Omega^4 (\Omega^2)'') f_r(r)
\eeaa
Therefore the expression \eqref{expression-to-be-positive} becomes
\beaa
&&\frac{1}{r^{11}}\big(656 Q^8 - 1822 M Q^6 r + 160 M^2 Q^4 r^2 + 
 1504 Q^6 r^2 + 2412 M^3 Q^2 r^3 - 2986 M Q^4 r^3 - 
 1152 M^4 r^4 \\
 &&+ 404 M^2 Q^2 r^4 + 1072 Q^4 r^4 + 
 744 M^3 r^5 - 1224 M Q^2 r^5 + 64 M^2 r^6 + 256 Q^2 r^6 - 
 104 M r^7 + 16r^8\big)
\eeaa
Plugging in $Q^2= \gamma M^2$ and $r=M( 1+\sqrt{1-\gamma})x$ we obtain the polynomial
\beaa
&&656 \gamma^4 - 1822 \gamma^3 ( 1+\sqrt{1-\gamma})x + 160  \gamma^2 ( 1+\sqrt{1-\gamma})^2x^2 + 1504 \gamma^3 ( 1+\sqrt{1-\gamma})^2x^2 \\
 &&+ 2412  \gamma ( 1+\sqrt{1-\gamma})^3x^3 - 2986  \gamma^2 ( 1+\sqrt{1-\gamma})^3x^3 - 1152 ( 1+\sqrt{1-\gamma})^4x^4 \\
 &&+ 404  \gamma ( 1+\sqrt{1-\gamma})^4x^4 + 1072 \gamma^2 ( 1+\sqrt{1-\gamma})^4x^4 + 744  ( 1+\sqrt{1-\gamma})^5x^5 - 1224  \gamma  ( 1+\sqrt{1-\gamma})^5x^5 \\
 &&+ 64 ( 1+\sqrt{1-\gamma})^6x^6 + 256 \gamma ( 1+\sqrt{1-\gamma})^6x^6 - 104  ( 1+\sqrt{1-\gamma})^7x^7 + 16( 1+\sqrt{1-\gamma})^8x^8
\eeaa
 Consider the case of $\gamma=0$. Then the above polynomial, up to dividing by $2^8 x^4$, reduces to
\beaa
16 x^4-52x^3+16x^2+93x-72 
\eeaa
The above polynomial evaluated at $x=1$ is equal to $1$, and therefore positive. We now prove that its derivative, $64 x^3 -156 x^2+32 x+93$ is positive for $x \geq 1$. The derivative has its minimum for $x \geq 1$ at $x=\frac{39+\sqrt{1137}}{48}$, where its value can be checked being positive. This implies positivity of the polynomial in $x \in [1, \infty]$ for $\gamma=0$. For $\gamma$ small enough, the above expression is therefore also positive.

Upon integrating \eqref{integration-identity-morawetz} with respect to the measure $du dv \sin d\th d\phi$ over any spacetime region of the form $[ u_0, u] \times [v_0, v] \times S_{u, v}$ with $f$ chosen as above, we see that we can estimate all boundary terms by the $T$ fluxes at $u$, $v$, $u_0$ and $v_0$. By the conservation of energy \eqref{degenerate-energy}, they can therefore be bounded by $F_{u_0}^T[\psi](v_0, v)+F_{v_0}^T[\psi](u_0, u)$. We estimate the bulk term by \eqref{bulk-est} and making use of conditions 1. and 2. and of the positivity of the angular term, we can obtain
\beaa
\int_{u_0}^{u} \int_{v_0}^{v}du dv\sin \th d\th d\phi &\Omega^2 \{\frac{1}{r^2}|R^*\phi|^2+\frac{1}{r^3}|\phi|^2+\frac{(r^2-3Mr+2Q^2)^2}{r^5}|\nabb\phi|^2\} \\
&\les  F^T_{u_0}[\phi](v_0, v)+  F^T_{v_0}[\phi](u_0, u)
\eeaa
Finally, a standard argument allows to recover the missing derivative in the above bulk estimate, integrating the relation above using an increasing function $f$, vanishing at third order at $r=r_P$. We finally get the following Morawetz estimate:
\bea\label{degenerate-morawetz}
\begin{split}
\int_{u_0}^{u} \int_{v_0}^{v}  du dv\sin \th d\th d\phi &\Omega^2 \{\frac{1}{r^2}|R^*\phi|^2+\frac{1}{r^3}|\phi|^2+\frac{(r^2-3Mr+2Q^2)^2}{r^5}\left(|\nabb\phi|^2+\frac{1}{r^2}|T\phi|^2 \right)\} \\
&\les  F^T_{u_0}[\phi](v_0, v)+  F^T_{v_0}[\phi](u_0, u)
\end{split}
\eea
Observe that the above bulk has a degeneration at the photon sphere caused by the trapping phenomen at $r=r_P$.

\subsection{Redshift estimates}\label{redshift-est}
Observe that the above bulk has a degeneration at the horizon, since it does not control all derivatives at $r=r_{\mathcal{H}}$. We can eliminate this degeneracy by making use of the redshift vectorfield. See \cite{Federico} for a derivation of the redshift estimate. 

Through the standard techniques of the redshift vectorfield, we can improve the energy conservation \eqref{degenerate-energy} to the following non-degenerate version:
\bea
F_u[\phi](v_0, v)+F_v[\phi](u_0, u)\les F_{v_0}[\phi](u_0, u)+F_{u_0}[\phi](v_0, v)
\eea
and the following improved version of the Morawetz estimate \eqref{degenerate-morawetz}:
\beaa
\int_{u_0}^{u} \int_{v_0}^{v} du dv\sin \th d\th d\phi &\Omega^2 \{\frac{1}{r^2}|R^*\phi|^2+\frac{1}{r^3}|\phi|^2+\frac{(r^2-3Mr+2Q^2)^2}{r^5}\left(|\nabb\phi|^2+\frac{1}{r^2}|\Omega \nabb_4 \phi|^2+\frac{1}{r^2}|\Omega^{-1} \nabb_3 \phi|^2 \right)\} \\
&\les  F_{u_0}[\phi](v_0, v)+  F_{v_0}[\phi](u_0, u)
\eeaa
Taking the limit $u, v \to \infty$, the left hand side is given by the degenerate Morawetz bulk $\mathbb{I}_{deg}[\phi]$. The two estimates above therefore give, for a $\phi$ as in Proposition \ref{estimate-l=1-pf-proposition} the following boundedness estimate:
\bea\label{non-deg-en}
\sup_u F_u[\phi](v_0, \infty) +\sup_v F_v[\phi](u_0, \infty) \les  F_{u_0}[\phi](v_0, \infty) + F_{v_0}[\phi](u_0, \infty) 
\eea
and the integrated decay estimate
\bea\label{non-deg-mor}
\mathbb{I}_{deg}[\phi]\les F_{u_0}[\phi](v_0, \infty) + F_{v_0}[\phi](u_0, \infty)
\eea

\subsection{The $r^p$ hierarchy estimates}\label{rp-est}
We will adapt the $r^p$ hierarchy estimates \cite{rp} to the Fackerell-Ipser equation, as done in \cite{DHR} and \cite{Federico}. 

Equation \eqref{fackerell} implies the following identity\footnote{See equation (296) in \cite{DHR}}, for $p$ and $k$ real numbers:
\bea\label{expression-rp-hierarchy}
\begin{split}
0=_{S}&\pr_u \Big(\frac{r^p}{\Omega^{2k}} |\Omega \nabb_4 \phi|^2 \Big)+\pr_v \Big(\frac{r^p}{\Omega^{2k-2}} | \nabb \phi|^2 \Big)+\pr_v \Big(\frac{r^p V}{\Omega^{2k}} |\phi|^2 \Big)-\pr_u \Big(\frac{r^p}{\Omega^{2k}}\Big) |\Omega \nabb_4 \phi|^2 -\pr_v \Big(\frac{r^p V}{\Omega^{2k}}\Big) |\phi|^2 \\
&+\Big( (2-p)r^{p-1} \Omega^{2-2k}+r^p (k-1) \Omega^{-2k+2} (\Omega^2)' \Big) |\nabb\psi|^2 
\end{split}
\eea
We integrate \eqref{expression-rp-hierarchy} for $1\leq p \leq 2$ with respect to the measure $du dv \sin \th d\th d\phi$ in a region of the form 
\beaa
\mathcal{R}=\{ (u, v) \in \mathcal{M} : r \geq R, u_0 \leq u \leq u_{fin}, v_0 \leq v \leq v_{fin} \}
\eeaa
for sufficiently large $R$, $u_{fin}$ and $v_{fin}$. 
We fix $R$ big enough so that the following holds in the region $\{ r\geq R \}$:
\beaa
-\pr_u \Big(\frac{r^p}{\Omega^{2k}}\Big) \geq \frac 1 2 r^{p-1}
\eeaa
for all $1 \leq p \leq 2$ and $k \leq 5$. 

We also calculate
\beaa
-\pr_v \Big(\frac{r^p V}{\Omega^{2k}}\Big)&=& -4Q^2\pr_v \Big(\frac{r^{p-4} }{\Omega^{2k-2}}\Big)=-4Q^2(p-4) r^{p-5}\Omega^{-2k+2} \pr_vr-4Q^2(1-k)r^{p-4}\Omega^{-2k} (\Omega^2)'\pr_vr \\
&=&4Q^2\frac{\pr_vr}{\Omega^{2k}}\left((4-p) r^{p-5}\Omega^{2}+(k-1)r^{p-4}(\Omega^2)'\right)
\eeaa
Recall that 
\beaa
\Omega^2&=& 1-\frac{2M}{r}+\frac{Q^2}{r^2}, \qquad (\Omega^2)'=\frac{2M}{r^2}-\frac{2Q^2}{r^3}
\eeaa
The above computation becomes 
\beaa
-\pr_v \Big(\frac{r^p V}{\Omega^{2k}}\Big)&=&4Q^2\frac{\pr_vr}{\Omega^{2k}}\left((4-p) r^{p-5}+2M r^{p-6}(k+p-5)+Q^2 r^{p-7}(-2k-p+6) \right)
\eeaa
This implies that given any $1\leq p \leq 2$, the choice $k=5$ ensures that 
\beaa
-\pr_v \Big(\frac{r^p V}{\Omega^{2k}}\Big) & \geq& 2 r^{p-6}
\eeaa
holds in $\mathcal{R}$ for sufficiently large $R$. Integrating \eqref{expression-rp-hierarchy} in the spacetime region $\mathcal{R}$ for $p=2-\de$, and using the Morawetz estimate \eqref{non-deg-mor} to bound the terms on the timelike hypersurface $r=R$, we obtain
\bea\label{rp-intermediate-estimate}
\begin{split}
&\int_{\mathcal{R}} du dv \sin\th d\th d\phi \Big[r|\Omega\nabb_4 \phi|^2 +r^{1-\de} |\nabb\phi|^2 +r^{-1-\de}|\phi|^2\Big]\\
&\leq C\int_{v_0}^{\infty}  dv \sin\th d\th d\phi \Big[r^2 |\Omega \nabb_4\phi|^2\Big](u_0, v)+ C( F_{u_0}[\phi](v_0, \infty) + F_{v_0}[\phi](u_0, \infty))
\end{split}
\eea
where we applied the Poincar{\'e} inequality.

Taking the limit as $u_{fin}, v_{fin} \to \infty$ and summing \eqref{rp-intermediate-estimate} to  \eqref{non-deg-mor} and to \eqref{non-deg-en}, we obtain estimate \eqref{total-estimate-l=1} of Proposition \ref{estimate-l=1-pf-proposition}.

\subsection{Higher order estimates}\label{higher-order-est}
We note that the Fackerell-Ipser equation \eqref{fackerell} trivially commutes with Lie differentiation with respect to the Killing fields on Reissner-Nordstr{\"o}m spacetime. Recalling the definition of the higher order energies \eqref{definition-higher-order}, we immediately obtain estimate \eqref{higher-order-l=1} of Proposition \ref{estimate-l=1-pf-proposition}.

Exploiting the $r^p$-hierarchy in the standard way, polynomial decay estimates can be obtained for $\phi$ (see for example \cite{DHR} and \cite{Federico}). 

\section{Estimates for $\tilde{\b}$ and $\tilde{\bb}$}

We prove here the Main Theorem. The proof is divided in two parts. First we estimate the projection to $\ell=1$ spherical mode of $\tilde{\b}$ and $\tilde{\bb}$, and then we estimate the projection to the $\ell\geq 2$ modes. Finally, we combine the two to prove the Theorem.

\subsection{Estimates for the $\ell=1$ spherical mode of $\tilde{\b}$ and $\tilde{\bb}$}\label{mode-l=1-b}
In this section we make use of the differential relation between $\tilde{\b}$ and $\pf$ \eqref{quantities-1} together with the estimates obtained for the projection to the $\ell=1$ mode of $\pf$ in Proposition \ref{estimate-l=1-pf-proposition}.

We summarize the estimates in the following proposition. 

\begin{proposition} Let $\tilde{\b}$ and $\tilde{\bb}$ be solutions to the generalized Teukolsky equation of spin $\pm1$ respectively in Reissner-Nordstr{\"o}m spacetime for $|Q| {\color{red}\ll} M$. Then the following estimates hold:
\begin{enumerate}
\item weighted boundedness and integrated decay estimate for $\tilde{\b}_{\ell=1}$ and $\tilde{\bb}_{\ell=1}$:
\bea\label{estimate-l=1-tilde-b}
\begin{split}
\mathbb{F}[\phi, \tilde{\b}_{\ell=1}]+\mathbb{I}[\phi, \tilde{\b}_{\ell=1}]&\les \mathbb{F}_0[\phi, \tilde{\b}_{\ell=1}] \\
\mathbb{F}[\phi, \tilde{\bb}_{\ell=1}]+\mathbb{I}[\phi, \tilde{\bb}_{\ell=1}]&\les \mathbb{F}_0[\phi, \tilde{\bb}_{\ell=1}] 
\end{split}
\eea
\item higher order energy and integrated decay estimates for any integer $n\geq 0$:
\bea\label{estimate-higher-order-l=1-tilde-b}
\begin{split}
\mathbb{F}^{n, T, \nabb}[\phi, \mathcal{D}\tilde{\b}_{\ell=1}]+\mathbb{I}^{n, T, \nabb}[\phi, \mathcal{D}\tilde{\b}_{\ell=1}]&\les \mathbb{F}^{n, T, \nabb}_0[\phi, \mathcal{D}\tilde{\b}_{\ell=1}] \\
\mathbb{F}^{n, T, \nabb}[\phi, \mathcal{D}\tilde{\bb}_{\ell=1}]+\mathbb{I}^{n, T, \nabb}[\phi, \mathcal{D}\tilde{\bb}_{\ell=1}]&\les \mathbb{F}^{n, T, \nabb}_0[\phi, \mathcal{D}\tilde{\bb}_{\ell=1}] 
\end{split}
\eea
\end{enumerate}
\end{proposition}

We prove the proposition in the following two subsections, using transport estimates for $\tilde{\b}$ and $\tilde{\bb}$. 

\subsubsection{Transport estimates for the projection to $\ell=1$ mode of $\tilde{\b}$ and $\tilde{\bb}$}
We estimate the projection to $\ell=1$ mode of $\tilde{\b}$ and $\tilde{\bb}$ through some basic transport estimates.

\begin{lemma} Let $\tilde{\b}$ be a solution to the generalized Teukolsky equation of spin $+1$ in Reissner-Nordstr{\"o}m spacetime for $|Q|\ll M$. Then along any null hypersurface of constant $u$ 
\beaa
\int_{v_0}^{\infty} dv \sin \th d\th d\phi \  \Omega^{2} r^{10-\de} \left(|(r\divv\tilde{\b})_{\ell=1}|^2+ |(r\curll\tilde{\b})_{\ell=1}|^2\right) \les \mathbb{F}_0[\phi, \tilde{\b}_{\ell=1}]
\eeaa 
In addition,
\beaa
\int_{u_0}^{\infty} \int_{v_0}^{\infty} du dv \sin \th d\th d\phi \  \Omega^{4} r^{9-\de} \left(|(r\divv\tilde{\b})_{\ell=1}|^2+ |(r\curll\tilde{\b})_{\ell=1}|^2\right) \les \mathbb{F}_0[\phi, \tilde{\b}_{\ell=1}]
\eeaa
Let $\tilde{\bb}$ be a solution to the generalized Teukolsky equation of spin $-1$ in Reissner-Nordstr{\"o}m spacetime for $|Q|{\color{red}\ll} M$. Then along any null hypersurface of constant $v$ 
\beaa
\int_{u_0}^{\infty} du \sin \th d\th d\phi \ r^{8} \left(|(r\divv\tilde{\bb})_{\ell=1}|^2+ |(r\curll\tilde{\bb})_{\ell=1}|^2\right) \les \mathbb{F}_0[\phi, \tilde{\bb}_{\ell=1}]
\eeaa 
In addition,
\beaa
\int_{u_0}^{\infty} \int_{v_0}^{\infty} du dv \sin \th d\th d\phi \   r^{7-\de} \left(|(r\divv\tilde{\bb})_{\ell=1}|^2+ |(r\curll\tilde{\bb})_{\ell=1}|^2\right) \les \mathbb{F}_0[\phi, \tilde{\bb}_{\ell=1}]
\eeaa
\end{lemma}
\begin{proof}
By commuting the differential relation \eqref{quantities-1} with $r\DDd_1$ and projecting into $\ell=1$ spherical mode, we obtain 
\bea\label{relations-div-l=1}
r\divv \pf_{l=1} =\frac{1}{\kab} \nabb_3(r^5 \kab (r\divv\tilde{\b}_{\ell=1}), \qquad r\curll \pf_{\ell=1} =\frac{1}{\kab} \nabb_3(r^5 \kab (r\curll\tilde{\b}_{\ell=1})
\eea
and similarly commuting \eqref{quantities-2} for spin $-1$ quantities:
\bea\label{relations-div-l=-1}
r\divv \underline{\pf}_{\ell=1} =\frac{1}{\ka} \nabb_4(r^5 \ka (r\divv\tilde{\bb}_{\ell=1}), \qquad r\curll \underline{\pf}_{\ell=1} =\frac{1}{\ka} \nabb_4(r^5 \ka (r\curll\tilde{\bb}_{\ell=1})
\eea
From \eqref{relations-div-l=1}, we have 
\beaa
\nabb_3(r^{10} \kab^{2} |r\divv\tilde{\b}_{\ell=1}|^2)&=& 2r^5\kab^2 (r\divv\tilde{\b}_{\ell=1}) (r\divv \pf_{\ell=1})
\eeaa
Multiplying by $r^n$ and using that $\nabb_3 r=\frac 1 2 r\kab$, we have
\beaa
\nabb_3(r^{10+n} \kab^{2} |r\divv\tilde{\b}_{\ell=1}|^2)+\frac{n}{2}r^{10+n} (-\kab)^{3} |r\divv\tilde{\b}_{\ell=1}|^2&=& 2r^{5+n}\kab^2 (r\divv\tilde{\b}_{\ell=1}) (r\divv \pf_{l=1})\\
&\leq& \frac{n}{4} r^{10+n} (-\kab)^3|r\divv\tilde{\b}_{\ell=1}|^2+\frac{4}{n} r^n (-\kab)|r\divv\pf_{\ell=1}|^2
\eeaa
by Cauchy-Schwarz. This therefore simplifies to
\beaa
\nabb_3(r^{10+n} \kab^{2} |r\divv\tilde{\b}_{\ell=1}|^2)+\frac{n}{4}r^{10+n} (-\kab)^{3} |r\divv\tilde{\b}_{\ell=1}|^2&\leq& \frac{4}{n} r^n (-\kab)|r\divv\pf_{\ell=1}|^2
\eeaa
Recalling that $\nabb_3=\Omega^{-1} \pr_u$ and $\kab=-\frac{2\Omega}{r}$ and $\phi=r^2 \divv \pf_{\ell=1}$ or $\phi=r^2 \curll \pf_{\ell=1}$, we obtain 
\bea\label{intermediate=1}
\begin{split}
 \pr_u(r^{8+n} \Omega^{2} |r\divv\tilde{\b}_{\ell=1}|^2)+\frac{n}{2} r^{7+n} \Omega^{4} |r\divv\tilde{\b}_{\ell=1}|^2&\leq \frac{2}{n} r^{n-3} \Omega^2|\phi|^2 \\
 \pr_u(r^{8+n} \Omega^{2} |r\curll\tilde{\b}_{\ell=1}|^2)+\frac{n}{2} r^{7+n} \Omega^{4} |r\curll\tilde{\b}_{\ell=1}|^2&\leq \frac{2}{n} r^{n-3} \Omega^2|\phi|^2 
 \end{split}
\eea
From \eqref{relations-div-l=-1}, we have 
\bea\label{intermediate-nabb-4-transport}
\nabb_4(r^{10} \ka^{2} |r\divv\tilde{\bb}_{\ell=1}|^2)&=& 2r^5\ka^2 (r\divv\tilde{\bb}_{\ell=1}) (r\divv \underline{\pf}_{\ell=1})
\eea
Multiplying by $r^{-\de}$ and using that $\nabb_4 r=\frac 1 2 r\ka$, we have
\beaa
\nabb_4(r^{10-\de} \ka^{2} |r\divv\tilde{\bb}_{\ell=1}|^2)+\frac{\de}{2}r^{10-\de} \ka^{3} |r\divv\tilde{\bb}_{\ell=1}|^2&=& 2r^{5-\de}\ka^2 (r\divv\tilde{\bb}_{\ell=1}) (r\divv \underline{\pf}_{\ell=1})\\
&\leq& \frac{\de}{4} r^{10-\de} \ka^3|r\divv\tilde{\bb}_{\ell=1}|^2+\frac{4}{\de} r^{-\de} \ka|r\divv\underline{\pf}_{\ell=1}|^2
\eeaa
which simplifies to
\beaa
\nabb_4(r^{10-\de} \ka^{2} |r\divv\tilde{\bb}_{\ell=1}|^2)+\frac{\de}{4}r^{10-\de} \ka^{3} |r\divv\tilde{\bb}_{\ell=1}|^2&\leq& \frac{4}{\de} r^{-\de} \ka|r\divv\underline{\pf}_{\ell=1}|^2
\eeaa
Recall that $\nabb_4=\Omega^{-1} \pr_v$ and $\ka=\frac{2\Omega}{r}$ and $\phi=r^2 \divv \underline{\pf}_{\ell=1}$ or $\phi=r^2 \curll \underline{\pf}_{\ell=1}$, we obtain 
\bea\label{intermediate-step-2}
\begin{split}
 \pr_v(r^{8-\de} \Omega^{2} |r\divv\tilde{\bb}_{\ell=1}|^2)+\frac{\de}{2}r^{7-\de} \Omega^{4} |r\divv\tilde{\bb}_{\ell=1}|^2&\leq \frac{2}{\de} r^{-3-\de} \Omega^2|\phi|^2 \\
 \pr_v(r^{8-\de} \Omega^{2} |r\curll\tilde{\bb}_{\ell=1}|^2)+\frac{\de}{2}r^{7-\de} \Omega^{4} |r\curll\tilde{\bb}_{\ell=1}|^2&\leq \frac{2}{\de} r^{-3-\de} \Omega^2|\phi|^2 
 \end{split}
\eea
We can also multiply \eqref{intermediate-nabb-4-transport} by $\frac{1}{\Omega^2}$ for which $\pr_v \Omega^{-2}=-\frac{1}{\Omega^2}\left( \frac{2M}{r^2}-\frac{2Q^2}{r^3}\right)$ and obtain
\beaa
\pr_v( r^{8} |r\divv\tilde{\bb}_{\ell=1}|^2)+\left( \frac{2M}{r^2}-\frac{2Q^2}{r^3}\right)r^{8}  |r\divv\tilde{\bb}_{\ell=1}|^2&=& 2r^3\Omega (r\divv\tilde{\bb}_{\ell=1}) (r\divv \underline{\pf}_{\ell=1})
\eeaa
and hence 
\bea\label{intermediate-again}
\pr_v( r^{8} |r\divv\tilde{\bb}_{\ell=1}|^2)+\left( M-\frac{Q^2}{r}\right)r^{6}  |r\divv\tilde{\bb}_{\ell=1}|^2&\leq& \left(\frac{r}{ Mr-Q^2}\right)r^{-2}\Omega^2|\phi|^2
\eea
The right hand side in \eqref{intermediate=1}, applied with $n=2-\de$, as well as the right hand side of \eqref{intermediate-step-2} and \eqref{intermediate-again} are estimated from initial data upon integration over a spacetime region $[ u_0, u] \times [v_0, v]\times S^2$ according to Proposition \ref{estimate-l=1-pf-proposition}, giving the desired bounds.
\end{proof}

Recalling the definitions of the energies and the bulks \eqref{definition-combined-esnergy} and \eqref{definition-combined-bulks}, the above Lemma implies estimate \eqref{estimate-l=1-tilde-b}.

\subsubsection{Higher derivative estimates}
We estimate higher derivative of the projection to the $\ell=1$ mode of $\tilde{\b}$ and $\tilde{\bb}$. 

Note that we already control the $\nabb_3 (r\divv\tilde{\b}_{\ell=1})$ and $\nabb_4(r\divv\tilde{\bb}_{\ell=1})$ by the differential relations \eqref{quantities-1} and \eqref{quantities-2} they satisfy. 

To estimate the remaining derivative, we commute the differential relations by $2R_*=-\Omega \nabb_3 +\Omega \nabb_4$. We compute
\beaa
\Omega \nabb_3 (R_{*}(r^4 \Omega (r\divv\tilde{\b}_{\ell=1}))=R_*(\Omega^2 \divv \pf_{\ell=1} )
\eeaa
Since the right hand side satisfies a non-degenerate estimate, we can use Proposition \ref{estimate-l=1-pf-proposition} to bound such term from initial data. 

In order to obtain the optimal weights in $r$, we commute the relation \eqref{relations-div-l=1} with $\Omega \nabb_4$. Indeed commuting
\beaa
\Omega\nabb_3(r^4 \Omega (r\divv\tilde{\b}_{\ell=1})=\Omega^2 r^{-2}\phi
\eeaa
by $\Omega \nabb_4$, we obtain
\beaa
\Omega \nabb_3[\Omega\nabb_4(r^4 \Omega (r\divv\tilde{\b}_{\ell=1})]=\Omega \nabb_4(\Omega^2 r^{-2}\phi)
\eeaa
Using a cut-off function which vanishes for $r \leq R$, we obtain the optima weights for the $\Omega \nabb_4$ derivative. This proves \eqref{estimate-higher-order-l=1-tilde-b} for $n=0$.

The result for $\tilde{\bb}_{\ell=1}$ is analogous, with the difference that there is no improvement in powers of $r$ when taking the $\nabb_3$ derivative.

Since the Fackerell-Ipser equation commutes with the Killing fields $T$ and angular momentum operators, the higher order estimates in  \eqref{estimate-higher-order-l=1-tilde-b} are implied. 

\subsection{Estimate for the $\ell\geq 2$ spherical modes of $\tilde{\b}$ and $\tilde{\bb}$}\label{lge2b}

In this section, we obtain estimates for $r \DDs_2 \tilde{\b}$ and $r \DDs_2 \tilde{\bb}$. We make use of the relations obtained in Lemma \ref{relations-derived}:
\beaa
r^3 \kab \DDs_2\tilde{\b}&=& -\rhoF\psi_1-\left(2\rhoF^2+3\rho \right)r^3\kab \ff, \\
r^3 \ka \DDs_2\tilde{\bb}&=& \rhoF\underline{\psi}_1-\left(2\rhoF^2+3\rho \right)r^3\ka \underline{\ff}
\eeaa
From the first relation, we can write 
\bea\label{relation-tilde-b-psi-1-ff}
  r\DDs_2\tilde{\b}&=& \frac{Q}{2\Omega r^3}\psi_1+\left(\frac{6M}{r^2}-\frac{8Q^2}{r^3} \right) \ff
\eea
This implies that 
\beaa
\int_{u_0}^{\infty} \int_{v_0}^{\infty} du dv \sin \th d\th d\phi \ i_{r \geq R}  |r \DDs_2\tilde{\b}|^2 r^n &\les& \int_{u_0}^{\infty} \int_{v_0}^{\infty} du dv \sin \th d\th d\phi \ i_{r \geq R} \left(r^{-6+n} |\psi_1|^2 +r^{-4+n} |\ff|^2\right)
\eeaa
The right hand side is controlled by the  initial data of $\qf$, $\qf^\F$, $\ff$, $\psi_1$ and $\a$ by Main Theorem of \cite{Giorgi4}, for $n=9-\de$. Recall that the theorem in \cite{Giorgi4} gives bound for the non-degenerate spacetime energies for the quantities $\psi_1$ and $\ff$, in Reissner-Nordstr{\"o}m spacetime with $|Q|\ll M$. Therefore, we can integrate \eqref{relation-tilde-b-psi-1-ff} in the whole spacetime region, and still control the right hand side by initial data. This proves 
\beaa
\mathbb{F}[\DDs_2\tilde{\b}]+\mathbb{I}[\DDs_2\tilde{\b}]&\les \mathbb{F}_0[\qf, \qf^\F, \ff, \psi_1, \a]+\mathbb{F}_0[\DDs_2\tilde{\b}]
\eeaa
For the spin $-2$ quantities, we have the relation 
\bea\label{relations-spin-1}
r  \DDs_2\tilde{\bb}&=& \frac{Q}{2\Omega r^3}\underline{\psi}_1+\left(\frac{6M}{r^2}-\frac{8Q^2}{r^3} \right) \underline{\ff}
\eea
from which we can similarly obtain
\beaa
\int_{u_0}^{\infty} \int_{v_0}^{\infty} du dv \sin \th d\th d\phi \ i_{r \geq R}  |r \DDs_2\tilde{\bb}|^2 r^n &\les& \int_{u_0}^{\infty} \int_{v_0}^{\infty} du dv \sin \th d\th d\phi \ i_{r \geq R} \left(r^{-6+n} |\underline{\psi}_1|^2 +r^{-4+n} |\underline{\ff}|^2\right)
\eeaa
The right hand side is controlled by the  initial data of $\underline{\qf}$, $\underline{\qf}^\F$, $\underline{\ff}$, $\underline{\psi}_1$ and $\aa$ by Main Theorem for spin $-2$ of \cite{Giorgi4}, for $n=7-\de$. As above, this gives
\beaa
\mathbb{F}[\DDs_2\tilde{\bb}]+\mathbb{I}[\DDs_2\tilde{\bb}]&\les \mathbb{F}_0[\underline{\qf}, \underline{\qf}^\F, \underline{\ff}, \underline{\psi}_1, \aa]+\mathbb{F}_0[\DDs_2\tilde{\bb}] 
\eeaa
Upon taking derivative with respect to $\nabb_3$, $\nabb_4$ and $\nabb$ of \eqref{relation-tilde-b-psi-1-ff} and integrating in spacetime, we can bound the left hand side by the initial data using Main Theorem in \cite{Giorgi4}. Indeed, both $\psi_1$ and $\ff$ are quantities obtained through transport estimates, therefore their spacetime norms are not degenerate. Observe all those derivatives improve of a power of $r$ in the estimates. This therefore implies 
\bea\label{final-estimate-lgeq2}
\mathbb{F}[\mathscr{D}\DDs_2\tilde{\b}]+\mathbb{I}[\mathscr{D}\DDs_2\tilde{\b}]&\les \mathbb{F}_0[\qf, \qf^\F, \ff, \psi_1, \a]+\mathbb{F}_0[\mathscr{D}\DDs_2\tilde{\b}]
\eea
In commuting \eqref{relations-spin-1} we obtain a similar estimate, observing that the energies for the spin $-2$ quantities $\underline{\psi}_1$ and $\underline{\ff}$ do not improve in powers of $r$. This implies as above
\bea\label{final-estimate-lgeq2-2}
\mathbb{F}[\overline{\mathscr{D}}\DDs_2\tilde{\bb}]+\mathbb{I}[\overline{\mathscr{D}}\DDs_2\tilde{\b}]&\les \mathbb{F}_0[\underline{\qf}, \underline{\qf}^\F, \underline{\ff}, \underline{\psi}_1, \aa]+\mathbb{F}_0[\overline{\mathscr{D}}\DDs_2\tilde{\bb}]
\eea
By commuting the relations \eqref{relation-tilde-b-psi-1-ff} and \eqref{relations-spin-1} by the KIlling fields $T$ and angular momentum operator, we trivially obtain the equivalent estimates for higher order derivatives.

\subsection{Proof of the Main Theorem}\label{proof-main-theorem}

We are now ready to prove the Main Theorem. The proof is a straightforward application of the elliptic estimates in Section \ref{section-elliptic-estimates}, and the estimates obtained in Section \ref{mode-l=1-b} for the projection to the $\ell=1$ spherical mode and in Section \ref{lge2b} for the projection to the $\ell\geq 2$ spherical harmonics.

By Corollary \ref{main-elliptic-estimate} applied to $\tilde{\b}$ and $\tilde{\bb}$, we have
\beaa
\int_{S} |\tilde{\b}|^2 \leq C \left(\int_{S} |r\divv \tilde{\b}_{\ell=1}|^2+|r\curll \tilde{\b}_{\ell=1}|^2 + |r\DDs_2 \tilde{\b}|^2 \right)
\eeaa
Recalling that $r\divv$, $r\curll$ and $r\DDs_2$ commute with $\nabb_3$ and $\nabb_4$, from the above elliptic estimate we can immediately infer that
\beaa
\mathbb{F}[\tilde{\b}] & \les& \mathbb{F}[\phi, \tilde{\b}_{\ell=1}]+\mathbb{F}[\phi, \mathcal{D}\tilde{\b}_{\ell=1}]+\mathbb{F}[\DDs_2 \tilde{\b}]+\mathbb{F}[\mathscr{D}\DDs_2 \tilde{\b}] \\
\mathbb{F}[\tilde{\bb}] & \les& \mathbb{F}[\phi, \tilde{\bb}_{\ell=1}]+\mathbb{F}[\phi, \mathcal{D}\tilde{\bb}_{\ell=1}]+\mathbb{F}[\DDs_2 \tilde{\bb}]+\mathbb{F}[\mathscr{D}\DDs_2 \tilde{\bb}]
\eeaa
and 
\beaa
\mathbb{I}[\tilde{\b}] & \les& \mathbb{I}[\phi, \tilde{\b}_{\ell=1}]+\mathbb{I}[\phi, \mathcal{D}\tilde{\b}_{\ell=1}]+\mathbb{I}[\DDs_2 \tilde{\b}]+\mathbb{I}[\mathscr{D}\DDs_2 \tilde{\b}] \\
\mathbb{I}[\tilde{\bb}] & \les& \mathbb{I}[\phi, \tilde{\bb}_{\ell=1}]+\mathbb{I}[\phi, \mathcal{D}\tilde{\bb}_{\ell=1}]+\mathbb{I}[\DDs_2 \tilde{\bb}]+\mathbb{I}[\mathscr{D}\DDs_2 \tilde{\bb}]
\eeaa
Combining \eqref{estimate-l=1-tilde-b}, \eqref{estimate-higher-order-l=1-tilde-b} and \eqref{final-estimate-lgeq2} and \eqref{final-estimate-lgeq2-2}, we prove estimate \eqref{estimate-l=1-tilde-b-total} of Main Theorem.

By applying the equivalent of the above estimates for higher derivative, we prove estimate \eqref{estimate-higher-order-l=1-tilde-b-total} of Main Theorem. 

As a standard application of the $r^p$ hierarchy estimates, we can prove the pointwise decay estimates \eqref{pointwise-decay}. The constant depends on two derivatives in $T$ of the initial data.

\appendix

\section{Derivation of the Teukolsky equations}\label{computations-appendix}

The computations in this appendix are done in full generality, in physical space, with no gauge assumption. To perform the computations, we use the equations in Section \ref{linearized-equations}.

\subsection{Preliminaries}
We collect here few lemmas of \cite{Giorgi4}.

\begin{lemma}[ Lemma A.1.1. of \cite{Giorgi4}]\label{square-RN} The wave operator of a $p$-rank $S$-tensor $\Psi$ is given by 
\bea
\Box_\g\Psi&=& -\frac 1 2 \nabb_{3}\nabb_{4} \Psi-\frac 1 2 \nabb_{4}\nabb_{3} \Psi +\lapp_n \Psi+\left(-\frac 12 \kab+\omb\right) \nabb_4\Psi+\left(-\frac 1 2 \ka+\om \right)\nabb_3\Psi+\left(\eta^C +\etab^C\right)\nabb_C \Psi \nn\\
&=& - \nabb_{3}\nabb_{4} \Psi +\lapp_n \Psi+\left(-\frac 12 \kab+2\omb\right) \nabb_4\Psi-\frac 1 2 \ka\nabb_3\Psi+2\eta^C\nabb_C \Psi, \label{formula-1-wave}\\
&=& - \nabb_{4}\nabb_{3} \Psi +\lapp_n \Psi+\left(-\frac 12 \ka+2\om\right) \nabb_3\Psi-\frac 1 2 \kab\nabb_4\Psi+2\etab^C\nabb_C \Psi \label{formula-2-wave}
\eea
where $\lapp_p\Psi=\g^{CD} \nabb_{C}\nabb_{D}\Psi$ is the Laplacian for $p$-tensors.
\end{lemma}

\begin{corollary}\label{connection-Federico} In double null gauge, we have for a $p$-tensor $\Psi$
\beaa
r\Omega^2\Box_\g \Psi&=&- \Omega\nabb_3(\Omega \nabb_4(r \Psi))+\Omega^2\lapp(r \Psi)+\Omega^2\rho (r\Psi)
\eeaa
\end{corollary}
\begin{proof} We compute 
\beaa
\Omega \nabb_3(\Omega \nabb_4(r \Psi))&=&\Omega \nabb_3 \left(\frac 1 2 \Omega r \ka \Psi+\Omega r \nabb_4\Psi\right)\\
&=&\frac 1 2\Omega (\nabb_3 \Omega) r \ka \Psi+\frac 1 2 \Omega^2 (\nabb_3 r) \ka \Psi+\frac 1 2 \Omega^2 r (\nabb_3\ka) \Psi+\frac 1 2 \Omega^2 r \ka \nabb_3\Psi\\
&&+\Omega (\nabb_3\Omega) r \nabb_4\Psi+\Omega^2 (\nabb_3r) \nabb_4\Psi+\Omega^2 r \nabb_3\nabb_4\Psi
\eeaa
Using \eqref{nabb-3-Omega} and \eqref{nabb-3-ka}, we obtain
\beaa
\Omega \nabb_3(\Omega \nabb_4(r \Psi))&=&\frac 1 2\Omega (-2\omb \Omega) r \ka \Psi+\frac 1 2 \Omega^2 (\frac 1 2 r\kab) \ka \Psi+\frac 1 2 \Omega^2 r (-\frac 1 2 \ka\kab+2\omb\ka+2\rho) \Psi+\frac 1 2 \Omega^2 r \ka \nabb_3\Psi\\
&&+\Omega (-2\omb \Omega) r \nabb_4\Psi+\Omega^2 (\frac 1 2 r\kab)\nabb_4\Psi+\Omega^2 r \nabb_3\nabb_4\Psi\\
&=& \Omega^2 r \left(\nabb_3\nabb_4\Psi+\left(\frac 1 2 \kab-2\omb \right)\nabb_4\Psi+\frac 1 2 \ka \nabb_3\Psi+\rho \Psi \right)
\eeaa
Using \eqref{formula-1-wave}, we easily obtained the desired relation.
\end{proof}

\begin{lemma}[Lemma A.1.3. of \cite{Giorgi4}]\label{wave-rescaled} Consider a rescaled tensor $r^n \kab^m \Psi$, for $n$ and $m$ two numbers. Then 
\bea
r^n \kab^m \nabb_3(\Psi)&=&\nabb_3(r^n \kab^m \Psi)+\left(\frac {m-n}{ 2}  \kab +2m \omb\right)r^n\kab^m \Psi, \label{rescaled-derivatives-r-kab}\\
r^n \kab^m \nabb_4(\Psi)&=&\nabb_4(r^n \kab^m \Psi)+\left(\frac {m-n}{ 2}  \ka -2m \om-2m \rho \kab^{-1}\right)r^n\kab^m  \Psi
\eea
Moreover, it verifies the following wave equation:
\beaa
\Box_\g (r^n \kab^m \Psi)&=& \big( -\frac{n(n+1)+m(m-1)-2nm}{4} \ka\kab +m(m-n-1)\om\kab+(m^2+2m-n-nm) \rho+2m \rhoF^2\\
&&-m(m-n-1) \omb\ka+4m(m+1) \om\omb+4m^2\rho\omb \kab^{-1}-2m \nabb_3\om\big)r^n \kab^m \Psi+r^n\kab^m \Box_\g( \Psi)\\
&&+\left(\frac {m-n} 2 \kab+2m\omb\right)r^n\kab^{m}\nabb_4( \Psi)+\left(\frac {m-n}{ 2} \ka-2m\om-2m\rho\kab^{-1}\right) r^n\kab^m\nabb_3(\Psi)
\eeaa
\end{lemma}

 We recall
\bea\label{commutator-DDd-Box}
(-r\DDd_1 \Box_1+\Box_0 r\DDd_1) \Phi &=& -K r\DDd_1 \Phi, \\
( -r\DDs_1 \Box_0+\Box_1 r\DDs_1)\phi&=&K r\DDs_1 \phi
\eea

We recall here the spin $+1$ Teukolsky equation for the electromagnetic component $\bF$. 
\begin{proposition}\label{Teukolsky-bF}[Spin $+1$ Teukolsky equation for $\bF$ - Proposition A.2.1 in \cite{Giorgi4}] Let $\mathcal{S}$ be a linear gravitational and electromagnetic perturbation around Reissner-Nordstr{\"o}m. Consider the (gauge-dependent) curvature component $\bF$, which is part of the solution $\mathcal{S}$. Then $\bF$ satisfies the following equation:
\beaa
\Box_\g \bF&=& -2\omb \nabb_4\bF+\left( \ka+2\om\right) \nabb_3\bF+\left(\frac 1 4 \ka\kab+\om\kab-3\omb \ka+\rhoF^2-2\nabb_4\omb\right) \bF \\
&& +\rhoF\left(2\divv \chih +4\b-2\nabb_3\xi+(\kab +8\omb)\xi \right) 
\eeaa
\end{proposition}

We recall here the fundamental Lemma in \cite{Giorgi4} to derive the wave equation for a quantity of the form $\underline{P}(\Psi)$. 

\begin{lemma}[Lemma B.0.2 in \cite{Giorgi4}]\label{generalwavePhi} 
Let $\Psi$ be a $p$-tensor which verifies a wave equation of the form: 
\bea\label{wave-eq-ABC}
\Box_\g \Psi &=& A\underline{P}^{-1}(\Psi)+B\Psi+C \underline{P}(\Psi) +M,
\eea
where $\underline{P}^{-1}(\Psi)$ is a $p$-tensor such that $P\underline{P}^{-1}(\Psi)=\Psi$, and $A$, $B$, $C$ are coefficients of the wave equation. Then the $p$-tensor $\underline{P}(\Psi)$ verifies the wave equation:
\beaa
\Box_\g(\underline{P} \Psi)&=& \left( \kab^{-1}r\nabb_3(A)+ r A\right)\underline{P}^{-1}(\Psi)+\left(\kab^{-1}r\nabb_3(B)+ r B+ A+\frac 1 2 r\rho+r\rhoF^2  \right)\ \Psi \\
&&+\left( B +\kab^{-1}r\nabb_3(C)+ r C+\frac 1 2 \ka\kab-4\rho-2\rhoF^2\right)\ \underline{P}(\Psi)+ \left( C+\frac{1}{r}\left( -\ka\kab+2\rho\right) \right)\ \underline{P}(\underline{P}(\Psi))\\
&&+\kab^{-1}r\nabb_3 M  +\frac 3 2r M
\eeaa
\end{lemma}

\subsection{Spin $+1$ Teukolsky-type equation for $\tilde{\b}$}

We derive here the spin $+1$ Teukolsky equation for the gauge invariant quantity $\tilde{\b}$.

\begin{proposition}\label{Teukolsky-tilde-b}[Generalized spin $+1$ Teukolsky equation in $l=1$ for $\tilde{\b}$] Let $\mathcal{S}$ be a linear gravitational and electromagnetic perturbation around Reissner-Nordstr{\"o}m. Consider the gauge-invariant curvature component $\tilde{\b}=2\rhoF \b-3\rho \bF$. Then $\tilde{\b}$ satisfies the following equation:
\beaa
\Box_\g (r^3 \tilde{\b})&=& -2\omb \nabb_4(r^3\tilde{\b})+\left(\ka+2\om\right) \nabb_3(r^3\tilde{\b})+\left(\frac {1}{4} \ka \kab-3\omb \ka+\om\kab-2\rho+3\rhoF^2-8\om\omb+2\nabb_3 \om\right) r^3 \tilde{\b}\\
&&-2r^3 \kab\rhoF^2 \left(\nabb_4 \bF+ (\frac 3 2 \ka+2\om)\bF-2\rhoF \xi \right)+8r^3 \rhoF^2 \divv \ff
\eeaa
\end{proposition}
\begin{proof} We compute the following, using \eqref{nabb-4-rhoF}, \eqref{nabb-4-b}, \eqref{nabb-4-rho}:
\beaa
\nabb_4 \tilde{\b} +(3\ka+2\om) \tilde{\b}&=& \nabb_4 \left(2\rhoF \b-3\rho \bF\right) +(3\ka+2\om) (2\rhoF \b-3\rho \bF)\\
&=&  2\nabb_4\rhoF \b+2\rhoF \nabb_4\b-3\nabb_4\rho \bF-3\rho \nabb_4\bF +(3\ka+2\om) (2\rhoF \b-3\rho \bF)\\
&=&  2(-\ka \rhoF) \b+2\rhoF (-(2\ka +2\om)\b +\divv \a+3\rho \xi +\rhoF (\nabb_4 \bF+2\om\bF-2\rhoF \xi))\\
&&-3(-\frac 3 2 \ka \rho - \ka \rhoF^2) \bF-3\rho \nabb_4\bF+(3\ka+2\om) (2\rhoF \b-3\rho \bF)\\
&=&2\rhoF \divv \a +\left(2\rhoF^2 -3\rho\right)\left(\nabb_4 \bF+ \left(\frac 3 2 \ka+2\om\right)\bF -2\rhoF \xi\right)
\eeaa
We therefore obtain
\bea
\nabb_4 \tilde{\b} &=&-(3\ka+2\om) \tilde{\b}+2\rhoF \divv \a +\left(2\rhoF^2 -3\rho\right)\left(\nabb_4 \bF+ \left(\frac 3 2 \ka+2\om\right)\bF -2\rhoF \xi\right), \label{nabb-4-tilde-b-general}
\eea
We apply $\nabb_3$ to \eqref{nabb-4-tilde-b-general} to derive $\nabb_3\nabb_4 \tilde{\b}$. We consider the following pieces:
\beaa
\nabb_3\left(-(3\ka+2\om) \tilde{\b} \right)&=&-3 \nabb_3\ka \tilde{\b}-2 \nabb_3\om \tilde{\b}-(3\ka+2\om)  \nabb_3\tilde{\b}=\left(\frac 3 2 \ka \kab-6\omb \ka-6\rho-2\nabb_3 \om\right) \tilde{\b}-\left(3\ka+2\om\right)  \nabb_3\tilde{\b}, \\
\nabb_3 \left( 2\rhoF \divv \a\right)&=& 2\nabb_3\rhoF \divv \a+2\rhoF \nabb_3\divv \a\\
&=&-2\kab\rhoF\divv \a+2\rhoF ( -\left( \kab-4\omb \right)\divv\a-2 \divv\DDs_2\, \b-3\rho \divv\chih \\
&&-2\rhoF \ \left(\divv\DDs_2\bF +\rhoF\divv\chih \right))\\
&=&2\rhoF\left(-\left( 2\kab-4\omb \right)\divv\a+ (\lapp_1+ K) \b+(-3\rho-2\rhoF^2) \divv\chih +\rhoF \ ( \lapp_1+ K)\bF\right)
\eeaa
In the latter expression, we can eliminate $\b$ by writing $2\rhoF \b=\tilde{\b}+3\rho \bF$, which gives 
\beaa
\nabb_3 \left( 2\rhoF \divv \a\right)&=&2\rhoF\left(-\left( 2\kab-4\omb \right)\divv\a+(-3\rho-2\rhoF^2) \divv\chih +\rhoF \ ( \lapp_1+ K)\bF\right)\\
&&+ (\lapp_1+ K) (\tilde{\b}+3\rho \bF)\\
&=&2\rhoF\left(-\left( 2\kab-4\omb \right)\divv\a+(-3\rho-2\rhoF^2) \divv\chih \right)+ (\lapp_1+ K) \tilde{\b}\\
&&+\left(2\rhoF^2+3\rho \right)  ( \lapp_1+ K)\bF
\eeaa
The derivative $\nabb_3$ applied to the last term of\eqref{nabb-4-tilde-b-general} gives two terms. The first one is:
\beaa
&&\nabb_3\left(2\rhoF^2 -3\rho\right)\left(\nabb_4 \bF+ \left(\frac 3 2 \ka+2\om\right)\bF -2\rhoF \xi\right)\\
&=& \left(-\kab\rhoF^2 +\frac 9 2 \kab \rho\right)(\nabb_4 \bF+ (\frac 3 2 \ka+2\om)\bF-2\rhoF \xi )\\
&=& -\kab\rhoF^2 \left(\nabb_4 \bF+ (\frac 3 2 \ka+2\om)\bF-2\rhoF \xi \right)+ \frac 3 2 \kab (3\rho)\left(\nabb_4 \bF+ (\frac 3 2 \ka+2\om)\bF-2\rhoF \xi \right)
\eeaa
We simplify the term multiplied by $\rho$, writing $\left( 3\rho\right)\left(\nabb_4 \bF+ \left(\frac 3 2 \ka+2\om\right)\bF -2\rhoF \xi\right)=-\nabb_4 \tilde{\b} -(3\ka+2\om) \tilde{\b}+2\rhoF \divv \a +\left(2\rhoF^2 \right)\left(\nabb_4 \bF+ \left(\frac 3 2 \ka+2\om\right)\bF -2\rhoF \xi\right)$, so it gives 
\beaa
\nabb_3\left(2\rhoF^2 -3\rho\right)\left(\nabb_4 \bF+ \left(\frac 3 2 \ka+2\om\right)\bF -2\rhoF \xi\right)&=& 2\kab\rhoF^2 \left(\nabb_4 \bF+ (\frac 3 2 \ka+2\om)\bF-2\rhoF \xi \right)\\
&& -\frac 3 2 \kab\nabb_4 \tilde{\b} +\left(-\frac 9 2 \ka\kab-3\om\kab\right) \tilde{\b}+\rhoF\left(3 \kab \divv \a \right)
\eeaa
The second term is given by
\beaa
&&\left(2\rhoF^2 -3\rho\right)\nabb_3\left(\nabb_4 \bF+ \left(\frac 3 2 \ka+2\om\right)\bF -2\rhoF \xi\right)\\
&=&(2\rhoF^2 -3\rho)(\nabb_3\nabb_4 \bF+ \frac 3 2 \nabb_3\ka\bF+2\nabb_3\om \bF+ (\frac 3 2 \ka+2\om) \nabb_3\bF -2\nabb_3 \rhoF \xi -2\rhoF \nabb_3 \xi)\\
&=&(2\rhoF^2 -3\rho)(\nabb_4\nabb_3 \bF+2\omb \nabb_4\bF+  (-\frac 3 4 \ka \kab+3\omb \ka+3\rho+2\nabb_3\om)\bF+ \frac 3 2 \ka \nabb_3\bF+2\kab \rhoF \xi -2\rhoF \nabb_3 \xi )
\eeaa
By Proposition \ref{Teukolsky-bF}, we have that
\beaa
\nabb_4\nabb_3 \bF&=&\left(-\frac 1 4 \ka\kab-\om\kab+3\omb \ka-\rhoF^2+2\nabb_4\omb\right) \bF+\left(-\frac 1 2 \kab+2\omb\right) \nabb_4\bF-\frac 3 2 \ka \nabb_3 \bF  + \lapp_1\bF\\
&& +\rhoF\left(-2\divv \chih -4\b+2\nabb_3\xi-\kab \xi-8\omb\xi \right)
\eeaa
 and writing $2\rhoF \b=\tilde{\b}+3\rho \bF$ as above, we have 
 \beaa
\nabb_4\nabb_3 \bF&=&\left(-\frac 1 4 \ka\kab-\om\kab+3\omb \ka-6\rho-\rhoF^2+2\nabb_4\omb\right) \bF+\left(-\frac 1 2 \kab+2\omb\right) \nabb_4\bF-\frac 3 2 \ka \nabb_3 \bF  + \lapp_1\bF\\
&& +\rhoF\left(-2\divv \chih+2\nabb_3\xi-\kab \xi-8\omb\xi \right)-2\tilde{\b}
\eeaa
We therefore obtain:
\beaa
&&\left(2\rhoF^2 -3\rho\right)\nabb_3\left(\nabb_4 \bF+ \left(\frac 3 2 \ka+2\om\right)\bF -2\rhoF \xi\right)\\
&=&(2\rhoF^2 -3\rho)\Big[\left(- \ka\kab-\om\kab+6\omb \ka-3\rho-\rhoF^2+2\nabb_4\omb+2\nabb_3\om\right) \bF+\left(-\frac 1 2 \kab+4\omb\right) \nabb_4\bF  + \lapp_1\bF\\
&& +\rhoF\left(-2\divv \chih+\kab \xi-8\omb\xi \right)-2\tilde{\b}  \Big]
\eeaa
The term $(2\rhoF^2 -3\rho)\left(-\frac 1 2 \kab+4\omb\right) \nabb_4\bF =\left(-\frac 1 2 \kab+4\omb\right)(2\rhoF^2 -3\rho)\nabb_4\bF$ can be simplified by using $\left(2\rhoF^2 -3\rho\right)\nabb_4 \bF=\nabb_4 \tilde{\b} +(3\ka+2\om) \tilde{\b}-2\rhoF \divv \a-\left(2\rhoF^2 -3\rho\right)\left((\frac 3 2 \ka+2\om)\bF -2\rhoF \xi\right) $, which finally gives, using \eqref{nabb-4-omb}
\beaa
&&\left(2\rhoF^2 -3\rho\right)\nabb_3\left(\nabb_4 \bF+ \left(\frac 3 2 \ka+2\om\right)\bF -2\rhoF \xi\right)\\
&=&(2\rhoF^2 -3\rho)\Big[\left(- \ka\kab-\om\kab+6\omb \ka-3\rho-\rhoF^2+2\nabb_4\omb+2\nabb_3\om\right) \bF+ \lapp_1\bF\\
&& +\rhoF\left(-2\divv \chih+\kab \xi-8\omb\xi \right)-2\tilde{\b}  \Big]\\
&&+\left(-\frac 1 2 \kab+4\omb\right)(\nabb_4 \tilde{\b} +(3\ka+2\om) \tilde{\b}-2\rhoF \divv \a-\left(2\rhoF^2 -3\rho\right)\left((\frac 3 2 \ka+2\om)\bF -2\rhoF \xi\right))\\
&=&(2\rhoF^2 -3\rho)\Big[\left(- \frac 1 4 \ka\kab-\rho+\rhoF^2\right) \bF+ \lapp_1\bF +\rhoF\left(-2\divv \chih \right)-2\tilde{\b}  \Big]\\
&&+\left(-\frac 1 2 \kab+4\omb\right)\nabb_4 \tilde{\b} +\left(-\frac 3 2 \ka\kab-\om\kab+12\omb\ka+8\om\omb\right) \tilde{\b}+\rhoF((\kab-8\omb) \divv \a)
\eeaa
We can finally put all this together and compute $\nabb_3\nabb_4\tilde{\b}$, using \eqref{Gauss}:
\beaa
\nabb_3\nabb_4\tilde{\b}&=& \left(\frac 3 2 \ka \kab-6\omb \ka-6\rho-2\nabb_3 \om\right) \tilde{\b}-\left(3\ka+2\om\right)  \nabb_3\tilde{\b}\\
&&+2\rhoF\left(-\left( 2\kab-4\omb \right)\divv\a+(-3\rho-2\rhoF^2) \divv\chih \right)+ (\lapp_1+ K) \tilde{\b}+\left(2\rhoF^2+3\rho \right)  ( \lapp_1+ K)\bF\\
&&+2\kab\rhoF^2 \left(\nabb_4 \bF+ (\frac 3 2 \ka+2\om)\bF-2\rhoF \xi \right)-\frac 3 2 \kab\nabb_4 \tilde{\b} +\left(-\frac 9 2 \ka\kab-3\om\kab\right) \tilde{\b}+\rhoF\left(3 \kab \divv \a \right)\\
&&(2\rhoF^2 -3\rho)\Big[( \lapp_1+K)\bF +\rhoF\left(-2\divv \chih \right)-2\tilde{\b}  \Big]\\
&&+\left(-\frac 1 2 \kab+4\omb\right)\nabb_4 \tilde{\b} +\left(-\frac 3 2 \ka\kab-\om\kab+12\omb\ka+8\om\omb\right) \tilde{\b}+\rhoF((\kab-8\omb) \divv \a)\\
&=& \left(-\frac {19}{4} \ka \kab+6\omb \ka-4\om\kab-\rho-3\rhoF^2+8\om\omb-2\nabb_3 \om\right) \tilde{\b}-\left(3\ka+2\om\right)  \nabb_3\tilde{\b}+\left(-2 \kab+4\omb\right)\nabb_4 \tilde{\b}+ \lapp_1 \tilde{\b}\\
&&+2\rhoF^2\Big[2 ( \lapp_1+ K)\bF-4\rhoF\divv\chih +\kab \left(\nabb_4 \bF+ (\frac 3 2 \ka+2\om)\bF-2\rhoF \xi \right)\Big]
\eeaa
We can therefore compute the wave equation using \eqref{formula-1-wave} :
\beaa
\Box_\g\tilde{\b}&=& - \nabb_{3}\nabb_{4} \tilde{\b} +\lapp_1 \tilde{\b}+\left(-\frac 12 \kab+2\omb\right) \nabb_4\tilde{\b}-\frac 1 2 \ka\nabb_3\tilde{\b}\\
&=& \left(\frac {19}{4} \ka \kab-6\omb \ka+4\om\kab+\rho+3\rhoF^2-8\om\omb+2\nabb_3 \om\right) \tilde{\b}+\left(\frac 5 2 \ka+2\om\right)  \nabb_3\tilde{\b}+\left(\frac 3 2  \kab-2\omb\right)\nabb_4 \tilde{\b}\\
&&+2\rhoF^2\Big[-2 ( \lapp_1+ K)\bF+4\rhoF\divv\chih -\kab \left(\nabb_4 \bF+ (\frac 3 2 \ka+2\om)\bF-2\rhoF \xi \right)\Big]
\eeaa
Using Lemma \ref{wave-rescaled} for $n=3$, $m=0$, we compute $\Box_\g (r^3 \tilde{\b})$: 
\beaa
\Box_\g (r^3 \tilde{\b})&=& \big( -3\ka\kab -3 \rho\big)r^3  \tilde{\b}+r^3 \Box_\g(\tilde{\b})-\frac 3 2 \kab r^3\nabb_4( \tilde{\b})-\frac {3}{ 2} \ka  r^3\nabb_3(\tilde{\b})\\
&=& \left(\frac {7}{4} \ka \kab-6\omb \ka+4\om\kab-2\rho+3\rhoF^2-8\om\omb+2\nabb_3 \om\right) r^3 \tilde{\b}+\left(\ka+2\om\right) r^3  \nabb_3\tilde{\b}-2\omb r^3\nabb_4  \tilde{\b}\\
&&+2r^3 \rhoF^2\Big[-2 ( \lapp_1+ K)\bF+4\rhoF\divv\chih -\kab \left(\nabb_4 \bF+ (\frac 3 2 \ka+2\om)\bF-2\rhoF \xi \right)\Big]
\eeaa
and using \eqref{rescaled-derivatives-r-kab}, we obtain
\beaa
\Box_\g (r^3 \tilde{\b})&=& -2\omb \nabb_4(r^3\tilde{\b})+\left(\ka+2\om\right) \nabb_3(r^3\tilde{\b})+\left(\frac {1}{4} \ka \kab-3\omb \ka+\om\kab-2\rho+3\rhoF^2-8\om\omb+2\nabb_3 \om\right) r^3 \tilde{\b}\\
&&+2r^3 \rhoF^2\Big[-2 ( \lapp_1+ K)\bF+4\rhoF\divv\chih -\kab \left(\nabb_4 \bF+ (\frac 3 2 \ka+2\om)\bF-2\rhoF \xi \right)\Big]
\eeaa 
Consider the term $-2 ( \lapp_1+ K)\bF+4\rhoF\divv\chih$. Using \eqref{angular-operators} and recalling the definition of the symmetric traceless $2$-tensor $\ff=\DDs_2 \bF+\rhoF \chih$ in \cite{Giorgi4}, we have
\bea\label{ff-in-terms-lapl-bF}
\begin{split}
-2 ( \lapp_1+ K)\bF+4\rhoF\divv\chih&= -2 ( -2\DDd_2 \DDs_2-K+ K)\bF+4\rhoF\DDd_2\chih\\
&= 4\DDd_2 \DDs_2 \bF+4\rhoF \DDd_2 \chih=4\divv \ff
\end{split}
\eea
which gives the desired expression.
\end{proof}

\begin{corollary}\label{square-psi5} The derived quantity $\psi_5=r^4\kab \tilde{\b}$ verifies the following wave equation:
\beaa
\Box_\g \psi_5&=& \left(-\frac{ 1}{ 4} \ka\kab+5 \rhoF^2\right)\psi_5+\frac 1 r \left(\ka\kab-2\rho\right) \psi_6\\
&&+ 2r^4 \kab\rhoF^2\Big[4\divv \ff -\kab \left(\nabb_4 \bF+ (\frac 3 2 \ka+2\om)\bF-2\rhoF \xi \right)\Big]
\eeaa
\end{corollary}
\begin{proof} Using Lemma \ref{wave-rescaled} for $\psi_5=r\kab (r^3 \tilde{\b})$ with $n=m=1$, we obtain 
\beaa
\Box_\g (\psi_5)&=& \big( -\om\kab+ \rho+2 \rhoF^2+ \omb\ka+8 \om\omb+4\rho\omb \kab^{-1}-2 \nabb_3\om\big)\psi_5+r\kab\Box_\g( r^3 \tilde{\b})\\
&&+\left(2\omb\right)r\kab\nabb_4(r^3\tilde{\b})+\left(-2\om-2\rho\kab^{-1}\right) r\kab\nabb_3(r^3\tilde{\b})\\
&=& \big( -\om\kab+ \rho+2 \rhoF^2+ \omb\ka+8 \om\omb+4\rho\omb \kab^{-1}-2 \nabb_3\om\big)\psi_5\\
&&+r\kab\Big[-2\omb \nabb_4(r^3\tilde{\b})+\left(\ka+2\om\right) \nabb_3(r^3\tilde{\b})+\left(\frac {1}{4} \ka \kab-3\omb \ka+\om\kab-2\rho+3\rhoF^2-8\om\omb+2\nabb_3 \om\right) r^3 \tilde{\b}\\
&&+2r^3 \rhoF^2\Big[4\divv \ff -\kab \left(\nabb_4 \bF+ (\frac 3 2 \ka+2\om)\bF-2\rhoF \xi \right)\Big]\Big]\\
&&+\left(2\omb\right)r\kab\nabb_4(r^3\tilde{\b})+\left(-2\om-2\rho\kab^{-1}\right) r\kab\nabb_3(r^3\tilde{\b})\\
&=& \left(\frac {1}{4} \ka \kab-  \rho+5 \rhoF^2-2 \omb\ka+4\rho\omb \kab^{-1}\right)\psi_5+r(\kab\ka-2\rho) \nabb_3(r^3\tilde{\b}) \\
&&+ 2r^4 \kab\rhoF^2\Big[4\divv \ff -\kab \left(\nabb_4 \bF+ (\frac 3 2 \ka+2\om)\bF-2\rhoF \xi \right)\Big]
\eeaa
and writing $r \kab \nabb_3(r^3\tilde{\b})=\nabb_3(\psi_5)+2 \omb \psi_5=\frac 1 r \kab \psi_6+(-\frac 1 2  \kab+2\omb) \psi_5$, we obtain
\beaa
\Box_\g (\psi_5)&=& \big(\frac{ 1}{ 4} \ka\kab - \rho+5 \rhoF^2-2 \omb\ka+4\rho\omb \kab^{-1}\big)\psi_5+(\ka-2\rho \kab^{-1}) (\frac 1 r \kab \psi_6+(-\frac 1 2  \kab+2\omb) \psi_5)\\
&&+ 2r^4 \kab\rhoF^2\Big[4\divv \ff -\kab \left(\nabb_4 \bF+ (\frac 3 2 \ka+2\om)\bF-2\rhoF \xi \right)\Big]\\
&=& \left(-\frac{ 1}{ 4} \ka\kab+5 \rhoF^2\right)\psi_5+\frac 1 r \left(\ka\kab-2\rho\right) \psi_6\\
&&+ 2r^4 \kab\rhoF^2\Big[4\divv \ff -\kab \left(\nabb_4 \bF+ (\frac 3 2 \ka+2\om)\bF-2\rhoF \xi \right)\Big]
\eeaa
as desired.
\end{proof}

\section{Derivation of the generalized Fackerell-Ipser equation for $\pf$ in $\ell=1$}\label{derivation-fackerell}
Recall that $\pf=\psi_6=\underline{P}(\psi_5)$. We make use of Lemma \ref{generalwavePhi} and Corollary \ref{square-psi5} to derive here the generalized Fackerell-Ipser equation for $\pf$.

\begin{proposition}\label{wave-eq-pf}[Generalized Fackerell-Ipser equation for $\pf$ in $\ell=1$] Let $\mathcal{S}$ be a linear gravitational and electromagnetic perturbation around Reissner-Nordstr{\"o}m. Consider the gauge-invariant derived quantity $\pf$. Then $\pf$ satisfies the following equation:
\beaa
\Box_\g\pf+\left(\frac 1 4 \ka\kab-5\rhoF^2\right)\pf&=&8r^2 \rhoF^2\divv(\qf^\F)
\eeaa
\end{proposition}
\begin{proof} By Corollary \ref{square-psi5}, we have that $\psi_5$ verifies a wave equation of the form \eqref{wave-eq-ABC} with $A=0$, $B=-\frac{ 1}{ 4} \ka\kab+5 \rhoF^2$, $C=\frac 1 r \left(\ka\kab-2\rho\right)$ and $M$ is given by 
\beaa
M&=& 2r^4 \kab\rhoF^2\Big[4\divv \ff -\kab \left(\nabb_4 \bF+ (\frac 3 2 \ka+2\om)\bF-2\rhoF \xi \right)\Big]
\eeaa
By Lemma \ref{generalwavePhi}, we compute
\beaa
&&\Box_\g(\pf)=\left(\kab^{-1}r\nabb_3(-\frac{ 1}{ 4} \ka\kab+5 \rhoF^2)+ r (-\frac{ 1}{ 4} \ka\kab+5 \rhoF^2)+\frac 1 2 r\rho+r\rhoF^2  \right)\ \psi_5 \\
&&+\left( -\frac{ 1}{ 4} \ka\kab+5 \rhoF^2+\kab^{-1}r\nabb_3(\frac 1 r \left(\ka\kab-2\rho\right))+ r \frac 1 r \left(\ka\kab-2\rho\right)+\frac 1 2 \ka\kab-4\rho-2\rhoF^2\right)\ \pf\\
&&+ \left( \frac 1 r \left(\ka\kab-2\rho\right)+\frac{1}{r}\left( -\ka\kab+2\rho\right) \right)\ \underline{P}\pf+\kab^{-1}r\nabb_3 M  +\frac 3 2r M\\
&=&\left(\kab^{-1}r(\frac 1 4 \ka\kab^2-\frac 1 2 \kab \rho-10\kab \rhoF^2)+ r (-\frac{ 1}{ 4} \ka\kab+5 \rhoF^2)+\frac 1 2 r\rho+r\rhoF^2  \right)\ \psi_5 \\
&&+\left( -\frac{ 1}{ 4} \ka\kab+5 \rhoF^2-\frac 1 2 (\ka\kab-2\rho)+\kab^{-1}(-\ka\kab^2+5\kab\rho+2\kab\rhoF^2)+ r \frac 1 r \left(\ka\kab-2\rho\right)+\frac 1 2 \ka\kab-4\rho-2\rhoF^2\right)\ \pf\\
&&+\kab^{-1}r\nabb_3 M  +\frac 3 2r M
\eeaa
which gives 
\beaa
\Box_\g\pf&=&\left( -\frac{ 1}{ 4} \ka\kab+5 \rhoF^2\right)\ \pf+\left(-4r\rhoF^2  \right)\ \psi_5+\kab^{-1}r\nabb_3 M  +\frac 3 2r M
\eeaa
We compute now the right hand side $\kab^{-1}r\nabb_3(M)+\frac 3 2 r M$. Recalling that $\rhoF=\frac{Q}{r^2}$, we write $M$ as
\beaa
M&=&Q^2\kab\Big( 8\divv \ff\Big) -2\kab^2Q^2 \left(\nabb_4 \bF+ \left(\frac 3 2 \ka+2\om\right)\bF -2\rhoF \xi\right) \\
&=& Q^2 (M_1+M_2)
\eeaa
with 
\beaa
M_1&=& 8\kab \ \divv \ff, \\
M_2&=&-2\kab^2 \left(\nabb_4 \bF+ \left(\frac 3 2 \ka+2\om\right)\bF -2\rhoF \xi\right) \\
\eeaa
We compute separately:
\beaa
\kab^{-1}r\nabb_3(M_1)+\frac 3 2 r M_1&=&\kab^{-1}r\nabb_3(8\kab \ \divv \ff)+12 r \kab \ \divv \ff=8\kab^{-1}r\nabb_3(\kab) \ \divv \ff+8\kab^{-1}r\kab \ \nabb_3(\divv \ff)+12 r \kab \ \divv \ff\\
&=&8\kab^{-1}r(-\frac 1 2 \kab^2-2\omb\kab) \ \divv \ff+8\kab^{-1}r\kab \ (\divv \nabb_3 \ff-\frac 1 2 \kab \divv \ff)+12 r \kab \ \divv \ff
\eeaa
Recall that\footnote{Equation (238) in \cite{Giorgi4}} 
\beaa
\nabb_3(\ff)+\left( \kab -2 \omb\right)\ff&=&  -\DDs_2\DDs_1(\rhoF, \sigmaF) -\frac 12\rhoF \left(  \kab \chih + \ka \chibh \right)
\eeaa
therefore we obtain
\beaa
\kab^{-1}r\nabb_3(M_1)+\frac 3 2 r M_1&=&r(4 \kab-16\omb) \ \divv \ff+8r \ \divv (-\left( \kab -2 \omb\right)\ff -\DDs_2\DDs_1(\rhoF, \sigmaF) -\frac 12\rhoF \left(  \kab \chih + \ka \chibh \right))\\
&=&-4 r\kab \ \divv \ff -8 r \DDd_2\DDs_2\DDs_1(\rhoF, \sigmaF) -4r\rhoF \left(  \kab \divv\chih + \ka \divv\chibh \right)
\eeaa
We now compute the second term:
\beaa
\kab^{-1}r\nabb_3(M_2)+\frac 3 2 r M_2&=& \kab^{-1}r\nabb_3(-2\kab^2   \left(\nabb_4 \bF+ \left(\frac 3 2 \ka+2\om\right)\bF -2\rhoF \xi \right))\\
&&+\frac 3 2 r (-2\kab^2   \left(\nabb_4 \bF+ \left(\frac 3 2 \ka+2\om\right)\bF -2\rhoF \xi\right))\\
&=& -4r\nabb_3(\kab)   \left(\nabb_4 \bF+ \left(\frac 3 2 \ka+2\om\right)\bF -2\rhoF \xi \right)\\
&&-2\kab  r\nabb_3  \left(\nabb_4 \bF+ \left(\frac 3 2 \ka+2\om\right)\bF -2\rhoF \xi\right)\\
&&+\frac 3 2 r (-2\kab^2   \left(\nabb_4 \bF+ \left(\frac 3 2 \ka+2\om\right)\bF -2\rhoF \xi \right))
\eeaa
which gives, using computations in Proposition \ref{Teukolsky-tilde-b}, 
\beaa
\kab^{-1}r\nabb_3(M_2)+\frac 3 2 r M_2&=& -4r(-\frac 1 2 \kab^2-2\omb\kab)   \left(\nabb_4 \bF+ \left(\frac 3 2 \ka+2\om\right)\bF -2\rhoF \xi \right)\\
&&-2\kab  r  \Big[\left(- \ka\kab-\om\kab+6\omb \ka-3\rho-\rhoF^2+2\nabb_4\omb+2\nabb_3\om\right) \bF+\left(-\frac 1 2 \kab+4\omb\right) \nabb_4\bF \\
&& + \lapp_1\bF+\rhoF\left(-2\divv \chih+\kab \xi-8\omb\xi \right)-2\tilde{\b}  \Big]\\
&&+\frac 3 2 r (-2\kab^2   \left(\nabb_4 \bF+ \left(\frac 3 2 \ka+2\om\right)\bF -2\rhoF \xi \right))\\
&=&-2\kab  r  \Big[ (\lapp_1+K)\bF+\rhoF\left(-2\divv \chih \right)-2\tilde{\b}  \Big]
\eeaa
Using again \eqref{ff-in-terms-lapl-bF}, we obtain
\beaa
\kab^{-1}r\nabb_3(M_2)+\frac 3 2 r M_2&=&4\kab  r  \left( \divv \ff+\tilde{\b}  \right)
\eeaa
Putting the two pieces together, we obtain 
\beaa
\kab^{-1}r\nabb_3(M)+\frac 3 2 r M&=& Q^2(-4 r\kab \ \divv \ff -8 r \DDd_2\DDs_2\DDs_1(\rhoF, \sigmaF) -4r\rhoF \left(  \kab \divv\chih + \ka \divv\chibh \right))\\
&&+ Q^2(4\kab  r  \left( \divv \ff+\tilde{\b}  \right))\\
&=& Q^2( 8r \divv(- \DDs_2\DDs_1(\rhoF, \sigmaF) -\frac 1 2 \rhoF \left(  \kab \chih + \ka \chibh \right)))+ Q^2(4\kab  r \tilde{\b}  )
\eeaa
Using that 
\beaa
\qf^\F&=& -r^3\DDs_2\DDs_1(\rhoF, \sigmaF) -\frac 12 r^3\rhoF \left(  \kab \chih + \ka \chibh \right)
\eeaa
we can finally write 
\beaa
\kab^{-1}r\nabb_3(M)+\frac 3 2 r M&=& \frac{Q^2}{r^2}( 8\divv(\qf^\F))+ Q^2(4\kab  r \tilde{\b}  )=8r^2 \rhoF^2\divv(\qf^\F)+4 r\rhoF^2 \psi_5
\eeaa
Finally the equation simplifies to 
\beaa
\Box_\g\pf&=&\left( -\frac{ 1}{ 4} \ka\kab+5 \rhoF^2\right)\ \pf+\left(-4r\rhoF^2  \right)\ \psi_5+8r^2 \rhoF^2\divv(\qf^\F)+4 r\rhoF^2 \psi_5\\
&=&\left( -\frac{ 1}{ 4} \ka\kab+5 \rhoF^2\right)\ \pf+8r^2 \rhoF^2\divv(\qf^\F)
\eeaa
as desired. 
\end{proof}

\begin{corollary}\label{wave-scalar-pf} The scalar quantity supported in $\ell=1$ $(\divv \pf)_{\ell=1}$ verifies the wave equation
\beaa
\Box_\g\left(r\divv\pf\right)_{\ell=1}-\left(\rho+4\rhoF^2\right)\left(r\divv\pf\right)_{\ell=1}&=&0
\eeaa
\end{corollary}
\begin{proof}
Using \eqref{commutator-DDd-Box} and \eqref{commute-wave}, we obtain
\beaa
\Box_\g(r\divv\pf)_{\ell=1}+(K+\frac 1 4 \ka\kab-5\rhoF^2)(r\divv\pf)_{\ell=1}&=&0
\eeaa
and using Gauss equation, we finally have the desired formula.
\end{proof}

\vspace{1cm}

\begin{flushleft}
\small{DEPARTMENT OF MATHEMATICS, COLUMBIA UNIVERSITY} \\
\textit{E-mail address}: {egiorgi@math.columbia.edu}
\end{flushleft}

\end{document}